\newif\iflong
\longtrue
\documentclass[a4paper,11pt]{article}
\usepackage[margin=1in]{geometry}
\usepackage{graphicx} 
\usepackage{authblk}
\title{Line Cover and Related Problems\thanks{This paper was presented at STACS 2026.}}
\author[1,2]{Matthias Bentert \thanks{Supported by the European Research Council (ERC) under the European Union’s Horizon 2020 research and innovation program (grant agreement No. 819416) and the German Research Foundation (DFG, grant 1-5003036: SPP 2378 -- ReNO-2).}}
\author[2]{Fedor V.  Fomin \thanks{ Supported by the Research Council of Norway under BWCA project, reference 314528,  and European Research Council (ERC) via grant NewPC, reference 101199930.}}
\author[2]{Petr A. Golovach \thanks{Supported by the Research Council of Norway under BWCA  (grant no.~314528) and Extreme-Algorithms (grant no~355137) projects.}}
\author[3]{Souvik Saha}
\author[3]{Sanjay Seetharaman}
\author[2]{Kirill Simonov}
\author[3]{Anannya Upasana}

\affil[1]{TU Berlin, Germany}
\affil[2]{ University of Bergen, Norway}
\affil[3]{The Institute of Mathematical Sciences, HBNI, India}
\usepackage{times}
\usepackage{soul}
\usepackage{url}
\usepackage[hidelinks]{hyperref}
\usepackage[utf8]{inputenc}
\usepackage[small]{caption}
\usepackage{graphicx}
\usepackage{amsmath}
\usepackage{amsthm}
\usepackage{booktabs}
\usepackage[numbers]{natbib}
\usepackage{etoolbox}
\usepackage[switch]{lineno}
\usepackage{graphicx} 
\usepackage{dsfont}
\usepackage{amsmath,amssymb}
\usepackage{amsthm}
\usepackage{todonotes}
\usepackage{xspace}
\usepackage{amssymb}
\usepackage{amstext}
\usepackage{amsmath}
\usepackage{amsthm,xpatch}
\usepackage{thmtools}
\usepackage{mathtools}
\usepackage{nicefrac}
\usepackage{hyperref}
\usepackage{cleveref}
\usepackage{thm-restate}
\DeclareMathOperator{\poly}{poly}
\newcommand{\Oh}{\mathcal{O}}

\newcommand{\bfp}{\mathbf{p}}

\newcommand{\bfx}{\mathbf{x}}
\newcommand{\bfy}{\mathbf{y}}
\newcommand{\norm}[1]{\left\lVert#1\right\rVert}
\newcommand{\bbR}{\mathbb{R}}
\DeclareMathOperator{\dist}{dist}
\newcommand{\pname}{\textsc}
\newcommand{\projcl}{\pname{Projective Clustering}\xspace}
\newcommand{\problemPLC}{\pname{Line Cover}\xspace}
\newcommand{\problemHLC}{\pname{Hyperplane Cover}\xspace}
\newcommand{\problemLineCl}{\pname{Line Clustering}\xspace}

\usetikzlibrary{patterns}

\newcommand{\many}{\ensuremath{n^{90}}}
\newcommand{\dlarge}{\ensuremath{d_{\ell}}}
\newcommand{\dsmall}{\ensuremath{d_{s}}}

\tikzstyle{point}=[circle,draw,inner sep=2pt]

\newcommand{\Rd}{\ensuremath{\mathds{R}^d}}

\usepackage{tcolorbox}

\newcommand{\defparprob}[4]{
\begin{tcolorbox}[colback=gray!5!white,colframe=gray!75!black]
  \vspace{-1mm}
  \begin{tabular*}{\textwidth}{@{\extracolsep{\fill}}lr} #1\\ 
  \end{tabular*}
  {\bf{Input:}} #2  \\
  {\bf{Question:}} #4
  \vspace{-1mm}
\end{tcolorbox}
}

\DeclareMathOperator{\operatorClassNP}{NP}
\newcommand{\classNP}{\ensuremath{\operatorClassNP}\xspace}

\DeclareMathOperator{\operatorClassNPH}{NP-Hard}
\newcommand{\classNPH}{\ensuremath{\operatorClassNPH}\xspace}

\DeclareMathOperator{\operatorClassFPT}{FPT\xspace}
\newcommand{\classFPT}{\ensuremath{\operatorClassFPT}\xspace}
\DeclareMathOperator{\operatorClassW}{W}
\newcommand{\classW}[1]{\ensuremath{\operatorClassW[#1]}}

\DeclareMathOperator{\operatorClassParaNP}{Para-NP\xspace}
\newcommand{\classParaNP}{\ensuremath{\operatorClassParaNP}\xspace}
\DeclareMathOperator{\operatorClassXP}{XP\xspace}
\newcommand{\classXP}{\ensuremath{\operatorClassXP}\xspace}
\DeclareMathOperator{\operatorClassETH}{ETH\xspace}
\newcommand{\classETH}{\ensuremath{\operatorClassETH}\xspace}

\newcommand{\calP}{\mathcal{P}}
\newcommand{\sgn}{{\sf{sign}}}
\newcommand{\xp}{{\sf{XP}}\xspace}

\newcommand{\pca}{{\sf{PCA}}\xspace}

\newcommand{\calC}{\mathcal{C}}
\newcommand{\comp}{{\sf{comp}}}

\newtheorem{lemma}{Lemma}

\newtheorem{theorem*}[lemma]{Theorem*}

\newtheorem{claim}{Claim}
\newtheorem{proposition}[lemma]{Proposition}
\newtheorem{reduction rule}{Reduction Rule}

\date{}
\begin{document}
	
	\maketitle
	
	\begin{abstract}
 \begin{abstract}
We study several extensions of the classic \problemPLC{} problem of covering a set of \(n\) points 
in the plane with   \(k\) lines.
\problemPLC{} is known to be \classNP-hard and our focus is on two natural generalizations: 
 (1)~\problemLineCl, where the objective is to find 
$k$ lines in the plane that minimize the sum of squares of distances of a given set of input points to the closest line, and 
(2)~\problemHLC, where the goal is to cover $n$ points in $\mathds{R}^d$ by $k$ hyperplanes. 
We also consider the more general \projcl{} problem, which unifies both of these 
and has numerous applications in machine learning, data mining, and computational geometry.
In this problem one seeks $k$ affine subspaces of dimension~$r$ minimizing the sum of squares of distances of a given set of $n$ points in \( \mathds{R}^d \) to the closest point within one of the~$k$ affine subspaces.

Our main contributions highlight notable differences in the parameterized complexity of these problems. While \problemPLC is fixed-parameter tractable with respect to the number~$k$ of lines in the solution, we show that \problemLineCl{} is \classW1-hard parameterized by~$k$, and that, under the Exponential Time Hypothesis, it admits no algorithm running in time \(n^{o(k)}\). Moreover, although \problemHLC{} has been known to be \classNP-hard since the 1980s, following work of Megiddo and Tamir, even for~$d=2$, we show that it remains NP-hard even when~$k=2$.

We complement our hardness results by presenting an algorithm for \projcl. We show that this problem is solvable in time \(n^{\mathcal{O}(dk(r+1))}\). This not only yields an upper bound for \problemLineCl{} that asymptotically matches our lower bound, but also significantly extends the benchmark exact algorithms for $k$-means clustering due to Inaba, Katoh, and Imai.

%
%
 \end{abstract} 
	\end{abstract}
	
\section{Introduction}

The \problemPLC   problem asks for the smallest number $k$ of lines that cover a given set of $n$ points in the Euclidean plane $\mathbb{R}^2$. This problem is  NP-hard and is  well-studied in the context of parameterized complexity, particularly because it admits a simple kernelization, often highlighted in introductory lectures on kernelization
\cite{cygan2015parameterized,fomin2019kernelization,Niedermeierbook06}. 

In this paper, we study two natural and well-established generalizations of \problemPLC, namely \problemLineCl and \problemHLC, along with their overarching generalization \projcl. Our results reveal an interesting contrast between the parameterized complexities of these problems and that of \problemPLC.

\problemLineCl is an optimization variant of \problemPLC, in which the goal is to find $k$ lines that best fit the points in the $\ell_2$ norm, rather than merely covering them.  In other words, we search for $k$ lines such that the sum of the squares of distances from each point to the closest point on one of the lines is minimized.
 
\defparprob{\problemLineCl}{
A set
$S=\{\bfx_1,\bfx_2, \dots, \bfx_n\}$ of points in \( \mathds{R}^2 \) and integers $k$ and $B$.}{$k$}{Are there $k$ lines  $L_1,L_2, \dots, L_k$ such that 
${\sum\limits_{i=1}^{n} \min\limits_{1\leq j\leq k}\dist(\bfx_i, L_j)^2\leq B}$?}
 \Cref{fig:xpalgoimg} provides an example of \problemLineCl with $k=3$.  
In Operations Research, 
 this  problem is known as \textsc{Line Facility Location}
 ~\cite{Megiddo1982}.
The problem has been studied for convex and piecewise linear functions~\cite{BOBROWSKI2012108} as well as in the weighted setting~\cite{DBLP:journals/jco/CheungD11}.

\begin{figure}[h]
    \centering
    \includegraphics[width=0.9\linewidth]{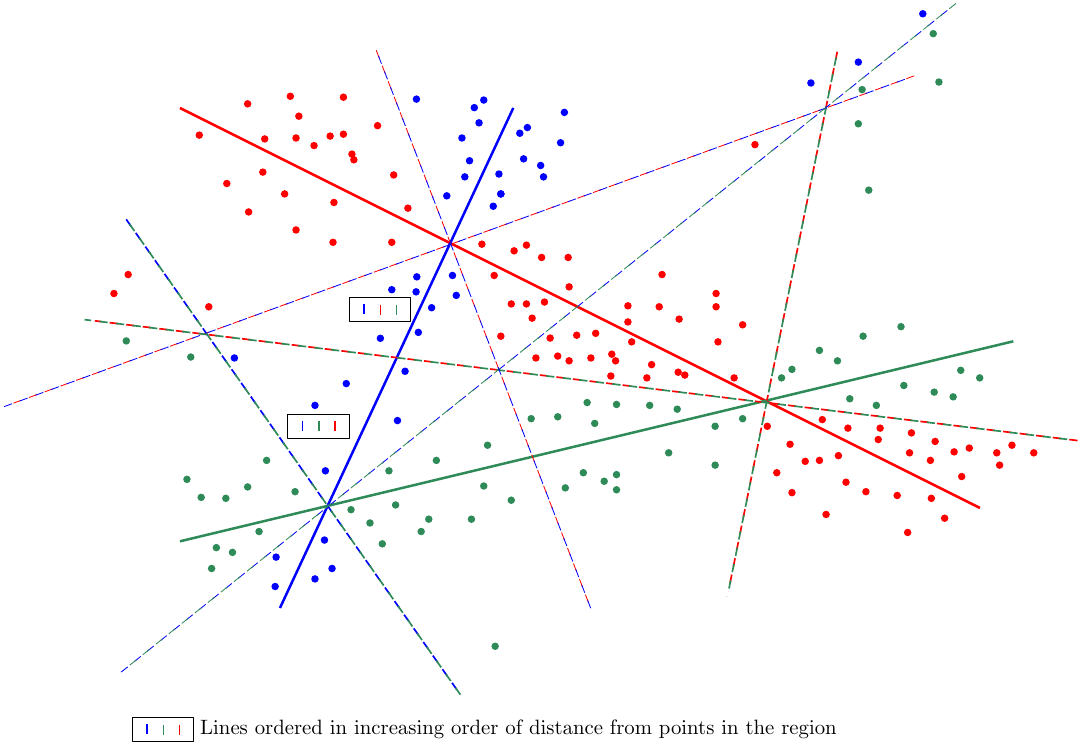}
    \caption{An instance of \problemLineCl{} with~$k=3$ and a solution (the three colored lines). The points are colored based on the color of the nearest line in the selected solution. The dashed lines show the region where points have equal distance from two solution lines.}
    \label{fig:xpalgoimg}
\end{figure}

The second natural generalization of  \problemPLC  that we consider in this paper is the extension to $d$-dimensional spaces.
In the  \problemHLC problem,  the task is to cover~$n$~points in $\mathds{R}^d$ by $k$ hyperplanes (affine subspaces of dimension $d-1$).
 
 \defparprob{\problemHLC}{
 A set
 $S=\{\bfx_1,\bfx_2, \dots, \bfx_n\}$ of points in \( \mathds{R}^d \) and an integer $k$.}{$k$}{Is it possible to cover all points of $S$ by $k$ hyperplanes? 
} 

Langerman and Morin~\cite{langerman2005covering} initiated the study of the parameterized complexity of \problemHLC\ (as part of a more general algorithmic framework). They showed that the problem is FPT when parameterized by $k + d$, providing several algorithms running in~$k^{\Oh(dk)} n^{\Oh(1)}$ time.

Finally, we consider the problem \projcl which extends both \problemLineCl and  \problemHLC, as well as many other well-studied problems. 
In this problem, we are given a family~$S$ of~$n$~points in~$\Rd$ and integers $r$ and $k$; note that we allow points to have the same coordinates.
The task is to find a set~$H$ of~$k$ affine subspaces of dimension~$r$ that minimizes the sum of squares of distances of each point in~$S$ to the closest point in the union of all subspaces in~$H$. 
It is also sometimes called \textsc{Affine Subspace Clustering} as the objective is to approximate points using a set of low-dimensional affine subspaces. The difference between linear and affine subspaces is that a linear subspace should always contain the origin of the coordinates. Affine subspaces can always be embedded into linear subspaces of one extra dimension. 
%

\defparprob{\projcl{}}{A multiset $S=\{\bfx_1,\bfx_2, \dots, \bfx_n\}$ of points in \( \mathds{R}^d \) and integers $r$, $k$, and~$B$.}{$k$}{Are there~$k$ $r$-dimensional affine subspaces $A_1,A_2, \dots, A_k$ whose union $H$ satisfies~$\sum\limits_{i=1}^{n}\dist(\bfx_i, H)^2 \leq B\label{eq:minsq1}$?}
Here, $\dist(\bfx_i, H)$ is the Euclidean distance of the point $\bfx_i$ to the closest point in one of the affine subspaces in~$H$. 
In particular, \problemLineCl is the special case with $d=2$ and $r=1$ and \problemHLC{} is the variant where~$r=d-1$ and~$B=0$.

\projcl{} naturally finds many applications in fields such as unsupervised learning~\cite{DBLP:reference/ml/Procopiuc10}, data mining~\cite{AggarwalPWYP99, parsons2004subspace}, artificial intelligence~\cite{10460289}, computational biology~\cite{DBLP:reference/ml/Procopiuc10}, database management~\cite{DBLP:conf/vldb/ChakrabartiM00}, computer vision~\cite{DBLP:conf/sigmod/ProcopiucJAM02}, image segmentation~\cite{yang2008unsupervised}, motion segmentation~\cite{vidal2008multiframe}, face clustering~\cite{ho2003clustering}, image processing~\cite{hong2006multi}, systems theory~\cite{vidal2003algebraic}, and others~\cite{elhamifar2013sparse, vidal2011subspace}.
It  also generalizes a number of fundamental and well-studied problems such as 
  \textsc{$k$-Means Clustering}~\cite{Lloyd82},  \textsc{Principal Component Analysis} (\pca)~\cite{Eckart1936TheAO}, and 
\textsc{Hyperplane Approximation}~\cite{korneenko1993hyperplane}.

\subparagraph{Our contribution.}
First, we show that \problemLineCl{} is \classW1-hard when parameterized by the number $ k$ of lines in the solution. Our reduction also implies a lower bound based on the Exponential Time Hypothesis (ETH).

\begin{restatable}{thm}{hardness}
    \label{thm:hardness}
   \problemLineCl is \classW1-hard parameterized by~$k$.  Assuming the \classETH, it cannot be solved in~$n^{o(k)}$ time.
\end{restatable}
We consider \Cref{thm:hardness} to be particularly interesting for several reasons. 
First, it underscores a clear contrast between \problemPLC, which is \classFPT when parameterized by $k$, and \problemLineCl. 
Second, \problemLineCl\ is closely related to $k$-means clustering, where instead of fitting $k$ lines, one seeks $k$ points (i.e., zero-dimensional subspaces) that best fit the data in the $L_2$ norm. 
Notably, the parameterized complexity of $k$-means clustering in $\mathds{R}^2$ (with parameter $k$) remains a longstanding open question \cite{cohen-addad2018bane}.

Our second lower bound concerns \problemHLC. Since in $\mathds{R}^2$ this problem is equivalent to \problemPLC, which is \classNP-complete, we know that the problem is \classParaNP-complete when parameterized by~$d$. On the other hand, by 
  the work of Langerman and Morin~\cite{langerman2005covering},  \problemHLC is \classFPT when parameterized by \( k \) and \( d \). 
\Cref{thm: objzerohard} refines  this complexity picture by proving that 
 \problemHLC
  is \classNPH even when \( k  = 2 \) (and~$d$ is large compared to~$k$). 

\begin{restatable}{thm}{objzerohard}
    \label{thm: objzerohard}
\problemHLC is \classNP-hard when $k=2$.
\end{restatable}


The lower bounds provided in \Cref{thm:hardness} rule out the existence of an algorithm for \problemLineCl  of running time 
$n^{o(k)}$. We complement this lower bound by providing an algorithm of running time $n^{O(k)}$. The algorithm is provided for the more general problem
 \projcl. To the best of our knowledge, this is the first  {\classXP} algorithm for \projcl{} parameterized by \( k \) and~\( d \).\footnote{As is common in computational geometry, we assume a real RAM model for computation.} Note that~$r \leq d$.


\begin{restatable}{thm}{xpalgo}
\label{theorem:algoforprojcl}
    \projcl{} can be solved in~$ n^{\mathcal{O}(\min\{dk(r+1),dk(d-r+1)\})}$ time.
\end{restatable} 
Since \textsc{\(k\)-Means Clustering} is the (very) special case of \(\projcl\) with~\(r=0\), 
\Cref{theorem:algoforprojcl} significantly extends the work of Inaba et al.~\cite{inaba1994applications}, who gave an algorithm for \textsc{\(k\)-Means Clustering} running in~\(n^{\Oh(kd)}\) time. It is also worth noting that the tools we use to prove 
\Cref{theorem:algoforprojcl} are quite different from those employed by Inaba et al., 
who relied on weighted Voronoi diagrams. In contrast, the proof of 
\Cref{theorem:algoforprojcl} leverages a fundamental result from computational algebraic geometry, 
which computes a set of sample points that meets every realizable sign condition of a family 
of polynomials restricted to an algebraic set.





\subparagraph{Related Work.}
\problemPLC{} is {\sf NP}-hard and {\sf APX}-hard \cite{kumar2000hardness,broden2001guarding,Megiddo1982}. 
A number of results in parameterized complexity and kernelization in the literature is devoted to  \problemPLC and \problemHLC. 
Langerman and Morin~\cite{langerman2005covering} proposed an \(k^{dk + d} n^{\Oh(1)}\)-time algorithm for \problemHLC. Later, the running time was improved by Wang et al.~\cite{wang2010parameterized} to~\(\tfrac{k^{(d-1)k}}{1.3^k} n^{\Oh(1)}\).  
Afshani et al.~\cite{AfshaniBDN16} gave an algorithm of running time~$2^n\cdot n^{\Oh(1)}$. They also improved on the 
 \classFPT algorithm for \problemHLC in \(\mathbb{R}^3\). Their improved algorithm runs in~\(\left(\tfrac{Ck^2}{\log k^{1/5}}\right)^k n^{\Oh(1)}\) time.
They also presented a \(\left(\tfrac{Ck}{\log k}\right)^kn^{\Oh(1)}\)-time algorithm for \problemPLC. 
 Kratsch et al. show that the well-known $\mathcal{O}(k^2)$ kernel for \problemPLC is tight under standard complexity assumptions~\cite{KratschPR14}. A similar result exists for \textsc{Hyperplane Cover}~\cite{BoissonnatDGK18}.

%
%
%

 Due to its importance in theory as well as applications, the study of \projcl{} enjoys a long and rich history. Various approaches have been used to solve \projcl{}: kernel sets~\cite{DBLP:journals/jacm/AgarwalHV04, DBLP:journals/siamcomp/Har-PeledW04},
coresets~\cite{DBLP:journals/siamcomp/Har-PeledW04, DBLP:conf/fsttcs/VaradarajanX12},
total sensitivity~\cite{DBLP:conf/stoc/FeldmanL11,DBLP:conf/soda/LangbergS10,DBLP:conf/fsttcs/VaradarajanX12}, and
dimension reduction via sampling~\cite{DBLP:journals/toc/DeshpandeRVW06, DBLP:journals/dcg/ShyamalkumarV12}.
\iflong
There is also a significant body of work on coresets constructions \cite{DBLP:conf/stoc/DeshpandeV07,har2004coresets, DBLP:conf/stoc/HuangV20} and dimension reductions \cite{CharikaW25}  for \projcl in computational geometry.
\fi

\section{Preliminaries}
  
 An $r$-dimensional affine subspace~$A$ (also called an~$r$-flat) of a $d$-dimensional Euclidean space~$\mathbb{R}^d$ is a subset of $\mathbb{R}^d$ of the form
 $
  \mathbf{p} + V=\{\mathbf{p} + \mathbf{v}: \mathbf{v}\in V\},$
 where $V$ is an $r$-dimensional linear subspace of $\mathbb{R}^d$
 and $\mathbf{p}$ is a fixed point in $\mathbb{R}^d$.
 In other words, an $r$-dimensional affine subspace~$A$  is a set of points 
\( \mathbf{x} =  (x_1, x_2, \dots, x_d)^T \in \mathbb{R}^d \)
  represented by the system of linear equations~$\mathbf{A} \mathbf{x} = \mathbf{p}$, where 
\( \mathbf{A} \in \mathbb{R}^{(d-r) \times d} \) is a matrix representing the coefficients of the equations and~\( \mathbf{p} \in \mathbb{R}^{(d-r) \times 1} \) is a vector of constants.


For a point or vector~$\bfx = (x_1,x_2,\ldots,x_d)^T$, we use~$\bfx[i]$ to denote the~$i$\textsuperscript{th} entry~$x_i$.
For two points $\bfx = (x_1, x_2, \dots, x_d)^T, \bfy= (y_1, y_2, \dots, y_d)^T  \in \bbR^d$, we denote the Euclidean distance between $\bfx$ and $\bfy$ by~$\dist(\bfx,\bfy) = \sqrt{\sum_{i = 1}^d (x[i]-y[i])^2}$. The Euclidean distance $\dist(\bfx,A)$  from a point $\bfx$ to an affine subspace $A$ is the distance between $\bfx$ and its orthogonal projection on the subspace $A$ (also known as the perpendicular distance).
Let $H$ be a union of $k$ affine subspaces $A_1,A_2, \dots, A_k$. For a point  $\bfx\in\bbR^d$, we define the distance   $\dist(\bfx,H)$ as the Euclidean distance from  $\bfx$ to the closest $A_i$, that is 
$
\dist(\bfx,H)=\min_{1\leq i\leq k}\dist(\bfx,A_i).$ 
We denote the square of the Frobenius norm of a matrix $A$ by  $\norm{\mathbf{A}}^2_F=\sum_{i,j}a_{i,j}^2$.


\medskip\noindent\textbf{Complexity.}
A \emph{parameterized problem} is a language $Q \subseteq \Sigma^* \times \mathbb{N}$, where $\Sigma^*$ is the set of strings over a finite alphabet $\Sigma$. Specifically, an input of $Q$ is a pair $(I, k)$, where~$I \in \Sigma^*$ and $k \in \mathbb{N}$; $k$ is the \emph{parameter} of the problem.
A parameterized problem $Q$ is \emph{fixed-parameter tractable} (\classFPT) if it can be decided whether $(I, k) \in Q$ in $f(k) |I|^{\Oh(1)}$ time for some computable function $f$ that depends only on the parameter $k$. The parameterized complexity class \classFPT consists of all fixed-parameter tractable problems.
 A parameterized problem $L \in \Sigma^* \times \mathbb{N}$ is called slice-wise polynomial {\sf (XP)} if there exists an algorithm $\mathcal{A}$ and two computable functions $f, g: \mathbb{N} \rightarrow \mathbb{N}$ such that, given an instance~$(I, k) \in \Sigma^* \times \mathbb{N}$, the algorithm $\mathcal{A}$ correctly decides whether $(x, k) \in L$ in~$f(k)\cdot |I|^{g(k)}$ time. The parameterized complexity class containing all slice-wise polynomial problems is called \classXP.

The $\sf{W}$-hierarchy is a collection of computational complexity classes; we omit the technical definitions here.
\iflong
The following relation is known among the classes in the $\sf{W}$-hierarchy:
\[
\classFPT = \classW0 \subseteq \classW1\subseteq\classW2 \subseteq \cdots \subseteq \classW P.
\]
\fi
It is widely believed that $\classFPT \neq \classW1$, and hence, if a problem is hard for a class $\classW i$ with $i \geq 1$, then it is considered to be fixed-parameter intractable.
For a detailed introduction to parameterized complexity, we refer readers to the textbooks by Cygan et al.\ and by Downey and Fellows~\cite{cygan2015parameterized,DowneyF99}.
We also provide conditional lower bounds by making use of  the following complexity hypothesis formulated by  Impagliazzo, Paturi, and Zane   \cite{ImpagliazzoPZ01}.
 
\begin{quote}
\textbf{Exponential Time Hypothesis (ETH)}:  There is a positive real $s$ such that 3-CNF-SAT with $n$ variables and $m$ clauses cannot be solved in time $2^{sn}(n+m)^{\Oh(1)}$.
\end{quote}

\section{\textsf{W}[1]-hardness of projective clustering in the Euclidean plane}

In this section, we show that \problemLineCl{} is {\sf \classW1}-hard when parameterized by the number~$k$ of lines in the solution.

\hardness*

\begin{proof}
    We give a reduction from \textsc{Regular Multicolored Independent Set}.
    Therein, one is given a graph~$G=(V,E)$ and a partition of the $n$ vertices into~$V_1,V_2,\ldots,V_\ell$ (often referred to as colors of vertices) such that~$|V_i| = \nu = \nicefrac{n}{\ell}$ for each~$i \in [\ell]$ and each vertex~$v$ has exactly~$q$ neighbors (none of which have the same color as~$v$).
    The question is whether there exists an independent set containing exactly one vertex of each color.
    It is known that \textsc{Regular Multicolored Independent Set} is \classW1-hard parameterized by~$\ell$~\cite[Proposition 13.5.]{cygan2015parameterized} and cannot be solved in~$n^{o(\ell)}$ time unless the \classETH fails \cite{Chen20061346}.
    We will assume without loss of generality that~$\nu$ is divisible by~$4$, 
    $\nu > \ell^3$, $\ell > 10$, and~$V_i = \{w_1^i,w_2^i,\ldots,w_\nu^i\}$.

    The crucial step is a polynomial gap-reduction that constructs an instance~$(S,k,B)$ of \problemLineCl{} where $S$ is a \emph{multiset} of points in~$\mathbb{R}^{2}$ (that is, we allow points to have the same coordinates) with integer coordinates and $k=2\ell+4$ with the following property:
    \begin{itemize}
    \item If the original instance of \textsc{Regular Multicolored Independent Set} is a yes-instance then there are $k$ lines such that the sum of squares of distances from all points in~$S$ to the closest line in the solution is at most~$B$.
    \item Otherwise, if we have a no-instance of \textsc{Regular Multicolored Independent Set} then for any $k$ lines, the sum of squares of distances from all points in~$S$ to the closest line is at least $B+2$.    
    \end{itemize}

    We use this gap-reduction to argue that \problemLineCl{} is \textsf{W}[1]-hard when $S$ is a set of points. For this, we apply the gap-reduction to an instance $G$ of \textsc{Regular Multicolored Independent Set} and denote by $(S,k,B)$ the obtained instance of \problemLineCl{} where $N$ is the number of points. Let $\delta=\nicefrac{1}{3BN}$. Then for every inclusion maximal multiset of points $X\subseteq S$ with the same coordinates $(x,y)$, we construct a set of $|X|$ points $Y_X$ with distinct rational coordinates such that every point $p\in Y_X$ is at distance at most $\delta$ from $(x,y)$. As $|X|\leq N$, $Y_X$ can be constructed in such a way that the denominators of the coordinates of each point are at most $\nicefrac{N}{\delta}=3BN^2$, and the absolute values of the numerators are at most 
    $N+3BN^2\max\{|x|,|y|\}$. Thus, the encoding of the coordinates of the point of $Y_X$ is polynomial. Notice also that the sets $Y_X$ constructed for distinct $X\subseteq S$ are disjoint. We denote by $S'$ the obtained set of distinct points $\mathbb{R}^2$, and let $B'=B+1$. 

    Consider the instance $(S',k,B')$ of \problemLineCl{}. Suppose that $G$ has a multicolored independent set. Then, there are $k$ lines such that the sum of squares of distances from all points in~$S$ to the closest line is at most~$B$. Suppose that a point $p\in S$ is on distance $h$ from the closest line. Because $p$ is at distance at most $\delta$ from the corresponding point $p'\in S'$, the distance from $p'$ to the closest line is at most $h+\delta$. Since the distance between each point of $S$ and the closest line is at most $B$,
    $(h+\delta)^2=h^2+2h\delta+\delta^2\leq h^2+2B\delta+\delta^2$. By taking the sum over all points in $S'$, we obtain that 
    the sum of squares of distances from the points in~$S'$ to the closest line is at most~$B+N\delta^2+2NB\delta=B+\nicefrac{1}{9B^2N}+\nicefrac{2}{3}<B+1=B'$. 
    Hence,     
    $(S',k,B')$ is a yes-instance. Assume now that $G$ has no multicolored independent set. Then, for any any $k$ lines, the sum of squares of distances from all points in~$S$ to the closest line is at least~$B+2$. If a point $p\in S$ is at distance $h$ from some line then the corresponding point $p'\in S'$ is at distance at least $h-\delta$ from the same line.
    For the square of the distance, we have that it is at least  
    $(h-\delta)^2=h^2-2h\delta+\delta^2\geq h^2-2B\delta+\delta^2\geq h^2-2B\delta$. 
    This implies that for any 
     $k$ lines, the sum of squares of distances from all points in~$S'$ to the closest line is 
     at least $(B+2)-2BN\delta>(B+2)-1=B+1=B'$. Thus, $(S',k,B')$ is a no-instance of \problemLineCl. This proves that  \problemLineCl{} is \textsf{W}[1]-hard when $S$ is a set of points. 
     
    In the remaining part of the proof, we provide the gap-reduction for a multiset of input points.
    Before we present the reduction in detail, we first give a high-level overview.
    Our aim is to build a grid-like structure where any optimal solution contains~$\ell$ vertical and~$\ell$ horizontal lines.
    Each horizontal line and each vertical line correspond to picking a vertex in a solution.
    The aim is then to construct a set of points such that (i) for each horizontal line in the solution there is a vertical line in the solution that encodes the same vertex (to ensure that a solution selects precisely~$\ell$ vertices) and (ii) for each pair of horizontal and vertical lines in a solution, the corresponding vertices are not adjacent in the input graph.
    We do this by placing points on the intersection of two such lines whenever the two respective vertices share an edge or are two distinct vertices of the same color.
    Otherwise, we place a point on the horizontal line close to the intersection.
    Imagine for now that these points have distinct distances from a ``solution line'' such that no vertical line contains more than one such point.
    See \Cref{fig:highlevel} for an example.
        \begin{figure}[t!]
        \centering
        \begin{tikzpicture}
            \def\x{3}
            \def\y{4}
            \node[circle,draw,fill=blue,label=left:$a$] (a) at (-6.5,1){};
            \node[circle,draw,fill=blue,label=left:$b$] (b) at (-6.5,0){};
            \node[circle,draw,fill=blue,label=left:$c$] (c) at (-6.5,-1){};
            \node[circle,draw,fill=red,label=right:$d$] at (-5.5,1){} edge(a);
            \node[circle,draw,fill=red,label=right:$e$] at (-5.5,0){} edge(c);
            \node[circle,draw,fill=red,label=right:$f$] at (-5.5,-1){} edge(b);

            \node[] at (-4.3,1.5){$a$};
            \node[] at (-4.3,1){$b$};
            \node[] at (-4.3,.5){$c$};
            
            \node[] at (-4.3,-1.5){$d$};
            \node[] at (-4.3,-2){$e$};
            \node[] at (-4.3,-2.5){$f$};
            
            \node[] at (-3.3,2.8){$a$};
            \node[] at (-2.3,2.8){$b$};
            \node[] at (-1.3,2.8){$c$};
            
            \node[] at (1,2.8){$d$};
            \node[] at (2,2.8){$e$};
            \node[] at (3,2.8){$f$};
            
            \draw[dashed] (-\y,-2.5) to (\y+.2,-2.5);
            \draw[dashed,red] (-\y,-2) to (\y+.2,-2);
            \draw[dashed] (-\y,-1.5) to (\y+.2,-1.5);
            
            \draw[dashed] (-\y,.5) to (\y+.2,.5);
            \draw[dashed] (-\y,1) to (\y+.2,1);
            \draw[dashed,blue] (-\y,1.5) to (\y+.2,1.5);

            \draw[dashed,blue] (-3.3,\x-.5) to (-3.3,-\x);
            \draw[dashed] (-2.3,\x-.5) to (-2.3,-\x);
            \draw[dashed] (-1.3,\x-.5) to (-1.3,-\x);
            
            \draw[dashed] (3,\x-.5) to (3,-\x);
            \draw[dashed,red] (2,\x-.5) to (2,-\x);
            \draw[dashed] (1,\x-.5) to (1,-\x);
            
            \node[point,fill=black] at(-3.1,1.5) {};
            \node[point,fill=black] at(-2.3,1.5) {};
            \node[point,fill=black] at(-1.3,1.5) {};
            \node[point,fill=black] at(-3.3,1) {};
            \node[point,fill=black] at(-2.1,1) {};
            \node[point,fill=black] at(-1.3,1) {};
            \node[point,fill=black] at(-3.3,.5) {};
            \node[point,fill=black] at(-2.3,.5) {};
            \node[point,fill=black] at(-1.1,.5) {};
            
            \node[point,fill=black] at(1,1.5) {};
            \node[point,fill=black] at(2.2,1.5) {};
            \node[point,fill=black] at(3.2,1.5) {};
            \node[point,fill=black] at(1.2,1) {};
            \node[point,fill=black] at(2.3,1) {};
            \node[point,fill=black] at(3,1) {};
            \node[point,fill=black] at(1.3,.5) {};
            \node[point,fill=black] at(2,.5) {};
            \node[point,fill=black] at(3.3,.5) {};
            
            \node[point,fill=black] at(-3.3,-1.5) {};
            \node[point,fill=black] at(-2,-1.5) {};
            \node[point,fill=black] at(-1,-1.5) {};
            \node[point,fill=black] at(-3,-2) {};
            \node[point,fill=black] at(-1.9,-2) {};
            \node[point,fill=black] at(-1.3,-2) {};
            \node[point,fill=black] at(-2.9,-2.5) {};
            \node[point,fill=black] at(-2.3,-2.5) {};
            \node[point,fill=black] at(-0.9,-2.5) {};
            
            \node[point,fill=black] at(1.4,-1.5) {};
            \node[point,fill=black] at(2,-1.5) {};
            \node[point,fill=black] at(3,-1.5) {};
            \node[point,fill=black] at(1,-2) {};
            \node[point,fill=black] at(2.4,-2) {};
            \node[point,fill=black] at(3,-2) {};
            \node[point,fill=black] at(1,-2.5) {};
            \node[point,fill=black] at(2,-2.5) {};
            \node[point,fill=black] at(3.4,-2.5) {};
        \end{tikzpicture}
        \caption{A simple input instance of \textsc{Regular Multicolored Independent Set} with~$\ell =2$, $\nu = 3$, and~$q=1$ on the left. The right shows a set of points such that any set of~$\ell$ vertical and~$\ell$ horizontal lines covering at least~$\ell^2\nu + \ell(\nu-1+q)=18$ points corresponds to a colorful independent set for the left instance. The colored dashed lines represent the independent set~$\{a,e\}$. Note that whenever a vertical and a horizontal line of the solution cross, there are no points on the intersection.}
        \label{fig:highlevel}
    \end{figure}
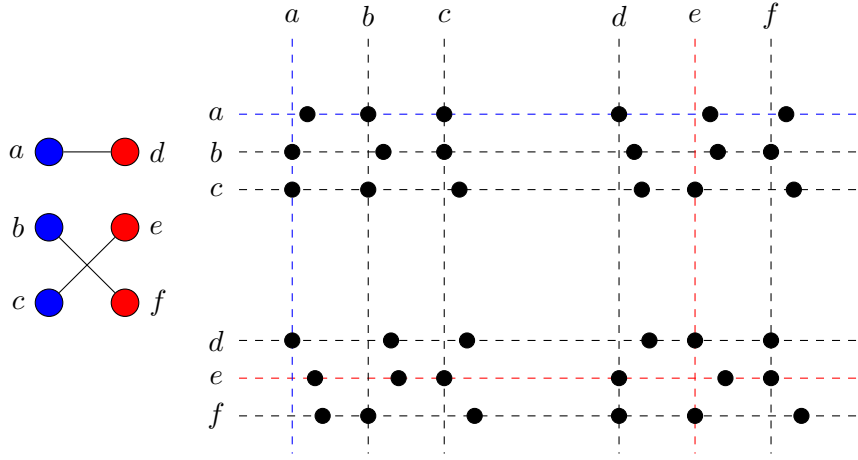
    Now, any independent set of size~$\ell$ containing exactly one vertex from each set~$V_i$ corresponds to a set of~$\ell$ horizontal and~$\ell$ vertical lines containing a maximum number of points.

    With this intuition in mind, there are two main obstacles to overcome.
    First, we want to ensure that the solution lines are actually (almost) horizontal and vertical and second, covering a maximum number of points is different from minimizing the sum of squares of distances from each point to the closest line.
    We start with the second obstacle.
    Assume that there are two lines (one vertical and one horizontal) that is part of any optimal solution that is far enough away from the previously described set of points so that they are never the closest line to any of these points.
    Let~$W$ be some sufficiently large (but also not too large) number.
    We place all vertical lines corresponding to vertices within the same set~$V_i$ closer together than lines corresponding to vertices in different sets (as already indicated in \Cref{fig:highlevel}).
    We do the same for all horizontal lines.
    By making the distance~$\dsmall$ between lines corresponding to vertices in different sets large enough, we can ensure that for any set~$V_i$ any optimal solution contains one horizontal line and one vertical line that corresponds to a vertex in~$V_i$ (assuming that all solution lines are horizontal or vertical).
    Next, we compute for each horizontal line the sum of squares of distances of all points on horizontal lines corresponding to vertices within the same set~$V_i$ except for those points that lie on the vertical line corresponding to the same vertex.
    For a line~$h$, let that number be~$\phi(h)$.
    Then, we place~$W-\phi(h)$ points on line~$h$ at a distance of one from the vertical line that is part of any optimal solution.
    We also place~$W$ points on each vertical line corresponding to picking a vertex at distance one from the horizontal line that is part of any solution and can now adjust the previous construction such that all points are covered by~$2\ell\nu$ vertical lines, that is, we no longer need to imagine that points placed because of edges in the input graph (or for pairs of vertices in the same set~$V_i$) have distinct distances.
    See \Cref{fig:highlevel2} for an illustration of the updated construction.
    \begin{figure}[t!]
        \centering
        \begin{tikzpicture}
            \def\x{3}
            \def\y{4}

            \node[] at (-4.3,1.5){$a$};
            \node[] at (-4.3,1){$b$};
            \node[] at (-4.3,.5){$c$};
            
            \node[] at (-4.3,-1.5){$d$};
            \node[] at (-4.3,-2){$e$};
            \node[] at (-4.3,-2.5){$f$};
            
            \node[] at (-3.3,2.8){$a$};
            \node[] at (-2.3,2.8){$b$};
            \node[] at (-1.3,2.8){$c$};
            
            \node[] at (1,2.8){$d$};
            \node[] at (2,2.8){$e$};
            \node[] at (3,2.8){$f$};
            
            \draw[dashed] (-\y,-2.5) to (\y+1.2,-2.5);
            \draw[dashed,red] (-\y,-2) to (\y+1.2,-2);
            \draw[dashed] (-\y,-1.5) to (\y+1.2,-1.5);
            
            \draw[dashed] (-\y,.5) to (\y+1.2,.5);
            \draw[dashed] (-\y,1) to (\y+1.2,1);
            \draw[dashed,blue] (-\y,1.5) to (\y+1.2,1.5);

            \draw[dashed,blue] (-3.3,\x-.5) to (-3.3,-\x-2);
            \draw[dashed] (-2.3,\x-.5) to (-2.3,-\x-2);
            \draw[dashed] (-1.3,\x-.5) to (-1.3,-\x-2);
            
            \draw[dashed] (3,\x-.5) to (3,-\x-2);
            \draw[dashed,red] (2,\x-.5) to (2,-\x-2);
            \draw[dashed] (1,\x-.5) to (1,-\x-2);

            \draw[violet] (-\y,-4.5) to (\y+1.1,-4.5);
            \draw[violet] (4.7,\x-.5) to (4.7,-\x-2);
            
            \draw[<->] (1.3,.8) to (2,.8);
            \node at(1.65,1.2) {$9$};
            
            \draw[<->] (1,.3) to (1.3,.3);
            \node at(1.15,.1) {$1$};

            \draw[<->] (0,1) to (0,.5);
            \node at(.2,.75) {$1$};

            \draw[<->] (4.7,-1.3) to (5,-1.3);
            \node at(4.85,-1) {$1$};
            
            \draw[<->] (-3.5,-4.8) to (-3.5,-4.5);
            \node at(-3.7,-4.65) {$1$};
            
            \draw[<->] (0,.-2) to (0,-4.5);
            \node at(.3,-3.25) {$\dsmall$};
            
            \draw[<->] (-2.3,-.5) to (2,-.5);
            \node at(0,-.2) {$\dsmall$};
            

            
            \node[point,fill=black] at(-3,1.5) {};
            \node[point,fill=black] at(-2.3,1.5) {};
            \node[point,fill=black] at(-1.3,1.5) {};
            \node[point,fill=black] at(-3.3,1) {};
            \node[point,fill=black] at(-2,1) {};
            \node[point,fill=black] at(-1.3,1) {};
            \node[point,fill=black] at(-3.3,.5) {};
            \node[point,fill=black] at(-2.3,.5) {};
            \node[point,fill=black] at(-1,.5) {};
            
            \node[point,fill=black] at(1,1.5) {};
            \node[point,fill=black] at(2.3,1.5) {};
            \node[point,fill=black] at(3.3,1.5) {};
            \node[point,fill=black] at(1.3,1) {};
            \node[point,fill=black] at(2.3,1) {};
            \node[point,fill=black] at(3,1) {};
            \node[point,fill=black] at(1.3,.5) {};
            \node[point,fill=black] at(2,.5) {};
            \node[point,fill=black] at(3.3,.5) {};
            
            \node[point,fill=black] at(-3.3,-1.5) {};
            \node[point,fill=black] at(-2,-1.5) {};
            \node[point,fill=black] at(-1,-1.5) {};
            \node[point,fill=black] at(-3,-2) {};
            \node[point,fill=black] at(-2,-2) {};
            \node[point,fill=black] at(-1.3,-2) {};
            \node[point,fill=black] at(-3,-2.5) {};
            \node[point,fill=black] at(-2.3,-2.5) {};
            \node[point,fill=black] at(-1,-2.5) {};
            
            \node[point,fill=black] at(1.3,-1.5) {};
            \node[point,fill=black] at(2,-1.5) {};
            \node[point,fill=black] at(3,-1.5) {};
            \node[point,fill=black] at(1,-2) {};
            \node[point,fill=black] at(2.3,-2) {};
            \node[point,fill=black] at(3,-2) {};
            \node[point,fill=black] at(1,-2.5) {};
            \node[point,fill=black] at(2,-2.5) {};
            \node[point,fill=black] at(3.3,-2.5) {};

            \node[point,fill=blue,label=right:$20$] at(\y+1,1.5) {};
            \node[point,fill=blue,label=right:$8$] at(\y+1,1) {};
            \node[point,fill=blue,label=right:$20$] at(\y+1,.5) {};
            
            \node[point,fill=blue,label=right:$20$] at(\y+1,-1.5) {};
            \node[point,fill=blue,label=right:$8$] at(\y+1,-2) {};
            \node[point,fill=blue,label=right:$20$] at(\y+1,-2.5) {};
            
            \node[point,fill=blue,label=below:$30$] at(-3.3,-4.8) {};
            \node[point,fill=blue,label=below:$30$] at(-2.3,-4.8) {};
            \node[point,fill=blue,label=below:$30$] at(-1.3,-4.8) {};
            
            \node[point,fill=blue,label=below:$30$] at(1,-4.8) {};
            \node[point,fill=blue,label=below:$30$] at(2,-4.8) {};
            \node[point,fill=blue,label=below:$30$] at(3,-4.8) {};

            \end{tikzpicture}
        \caption{The updated construction (for the same instance as in \Cref{fig:highlevel}). The violet lines represent lines that are part of any optimal solution and the blue dots with numbers next to them represent a collection of points at the same position.
        The value of~$\dsmall$ is at least~$500$.
        The colored dashed lines (together with the violet lines) represent an optimal solution corresponding to the independent set~$\{a,e\}$.
        The horizontal line representing vertex~$a$ is the closest line in the solution to~$4$ points each on the horizontal lines representing~$b$ and~$c$.
        Since the distance to points on the line for~$b$ is one, this adds a cost of four to the budget. Points on the line for~$c$ have distance~$2$, so each points contributes~$4$. Summed up over all of these points, this gives an additional budget of~$20$ which is balanced out by the blue vertex on the line representing~$a$. For the edge representing~$e$, there are 8 points at distance one (four on the lines for~$d$ and~$f$, respectively). Note that the blue point on the line for~$e$ represents~$8$ points. Figure not drawn to scale due to the orders of magnitude difference in the distances.}
        \label{fig:highlevel2}
    \end{figure}
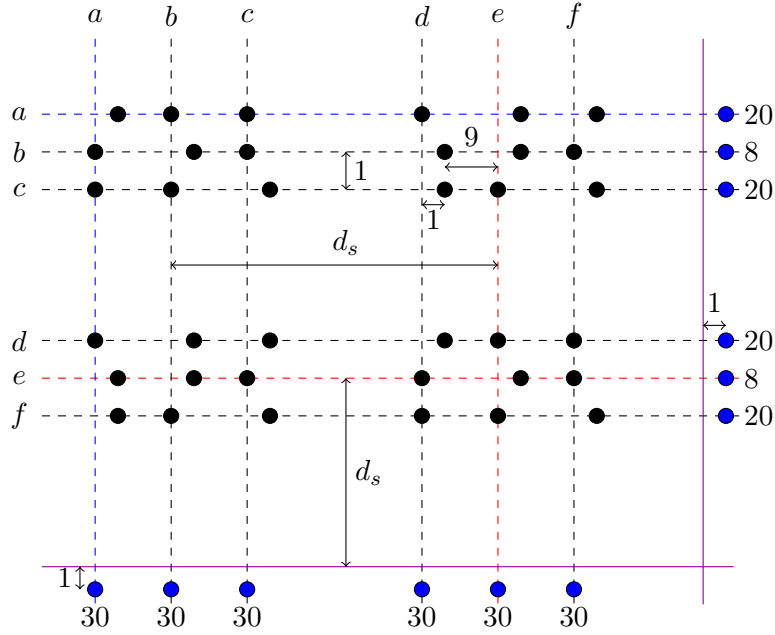
    We will show that these newly added points ensure that any set of~$\ell$ vertical and~$\ell$ horizontal lines that minimize the sum of squares of distances to all points are also lines that cover a maximum number of points.

    Finally, we want to ensure that lines in the solution are horizontal and/or vertical.
    Unfortunately, we were not able to ensure this property but the following construction ensures that all solution lines are \emph{almost} vertical or horizontal.
    Consider the construction in \Cref{fig:faraway} where~$\dlarge$ is much larger than~$\dsmall$.
    We will show that any line that covers as many points as there are points in one row or column of the constructed point set, has to contain the points of exactly one row or one column.
    We will then pick a suitable number~$p$, duplicate each point in the construction in \Cref{fig:highlevel2} $p$ times, and set the budget~$B$ correspondingly such that any line that encodes a vertex in some set~$V_i$ has a sum of squares of distances to points in the~$i$\textsuperscript{th} row or column of the construction in \Cref{fig:faraway} that is negligible while any line that is not almost vertical or horizontal has a point in any row or column of the construction such that the distance between this point and the line alone is larger than the budget~$B$.
    We also place a lot of vertices in the eight ``corners'' of the construction to ensure that there are two horizontal and two vertical lines that are part of any solution.
    This concludes the overview of the reduction.
    We now present the construction in more detail.

    We start by defining ${p=n^{10}}$, $W=n^{30}$, $\dsmall = n^{40}$, $\dlarge = n^{90}$, and the functions
    \begin{align*}
        \vartheta(i) &= \sum_{a=1}^{i} (3(i-a))^2 + \sum_{b=i}^{\nu} (3(\nu-b))^2, &
        \varphi(i) &= p\ell(\nu-1)\vartheta(i), \text{ and }
        \varphi'(i) = p \ell \nu \vartheta(i).
    \end{align*}
    We set~$k=2\ell+4$ and~$B = n^7 + (n-\ell) W + \ell \sum_{i=1}^\nu (W + \varphi(i)) - \ell W + \ell p(n - \nu + 1 - q - \ell)$.
    Note that~$B \leq n^{32}$.
    We then construct a set~$F$ of $8 \many + k^2+2k$ points as follows.
    We place points in a $\frac{k}{2} \times (k+2)$ grid~$G_h$ where the vertical distance between points is~$\dsmall$ and the horizontal distance is $\dlarge$.
    We leave $(\ell+1) \dsmall$ extra space between columns $\nicefrac{k+2}{2}$ and~$\nicefrac{k+2}{2}+1$ (this is where we will place the main part of the construction depicted in \Cref{fig:highlevel2}).
    We place~\many{}~points at each of the four corners of the grid (on top of the already existing point).
    We then repeat the process with a~$(k+2) \times \frac{k}{2}$ grid~$G_v$ where the vertical distance between points is $\dlarge$ and the horizontal distance is~$\dsmall$.
    We place the two grids such that the additional extra space forms a square in the middle of both grids.
    We call the set of all points placed so far~$F$ and the vertical/horizontal lines that each go exactly through $2 \many + k+2$ points the \emph{fixed lines} (we will show later that we can assume these four lines to be part of any optimal solution).
    An illustration of the points in $F$ is given in \Cref{fig:faraway}.
    \begin{figure}[t]
        \centering
        \begin{tikzpicture}
        \def\x{2.75}
            \foreach \i in {-2,-1,...,2}{
                \foreach \j in {-2,-1,...,2}{
                    \node[point] at(\i*.25,\j*.75-\x) {};
                    \node[point] at(\i*.25,\j*.75+\x) {};
                    \node[point] at(\i*.75-\x,\j*.25) {};
                    \node[point] at(\i*.75+\x,\j*.25) {};
                }
                \node[point] at(\i*.25,-5) {};
                \node[point] at(\i*.25,5) {};
                \node[point] at(-5,\i*.25) {};
                \node[point] at(5,\i*.25) {};
            }
            \node[point,fill=blue] at(-.5,-5) {};
            \node[point,fill=blue] at(.5,-5) {};
            \node[point,fill=blue] at(-.5,5) {};
            \node[point,fill=blue] at(.5,5) {};
            \node[point,fill=blue] at(-5,-.5) {};
            \node[point,fill=blue] at(5,-.5) {};
            \node[point,fill=blue] at(-5,.5) {};
            \node[point,fill=blue] at(5,.5) {};

            \draw[violet] (-5,-.5) to (5,-.5);
            \draw[violet] (-5,.5) to (5,.5);
            \draw[violet] (-.5,-5) to (-.5,5);
            \draw[violet] (.5,-5) to (.5,5);

            \draw (-.25,5) -- (0,5);
            \node at(-.125,5.3) {\dsmall};
            \draw (-2.75,-.25) -- (-2,-.25);
            \node at(-2.375,0) {\dlarge};
            \draw (-1.25,0) -- (-.5,0);
            \node at(-.875,.25) {\dlarge};
            \draw (-4.25,-.25) -- (-4.25,0);
            \node at(-4.6,-.125) {\dsmall};
            \draw (-.25,-2.75) -- (-.25,-2);
            \node at(0,-2.375) {\dlarge};
 
            \node[rectangle,draw,minimum width=.59cm,minimum height=.62cm,pattern=grid] at(0,0) {};
        \end{tikzpicture}
        \caption{An illustration of the construction of points in $F$ for~$\ell = 3$ (and~${k = 10}$). The blue dots represent $\many + 1$ points at the same coordinate. The four purple lines show the fixed lines and the rest of the construction happens in the box in the middle. Figure not drawn to scale as~$\dlarge$ is larger than~$\dsmall$ by many orders of magnitude.}
        \label{fig:faraway}
    \end{figure}
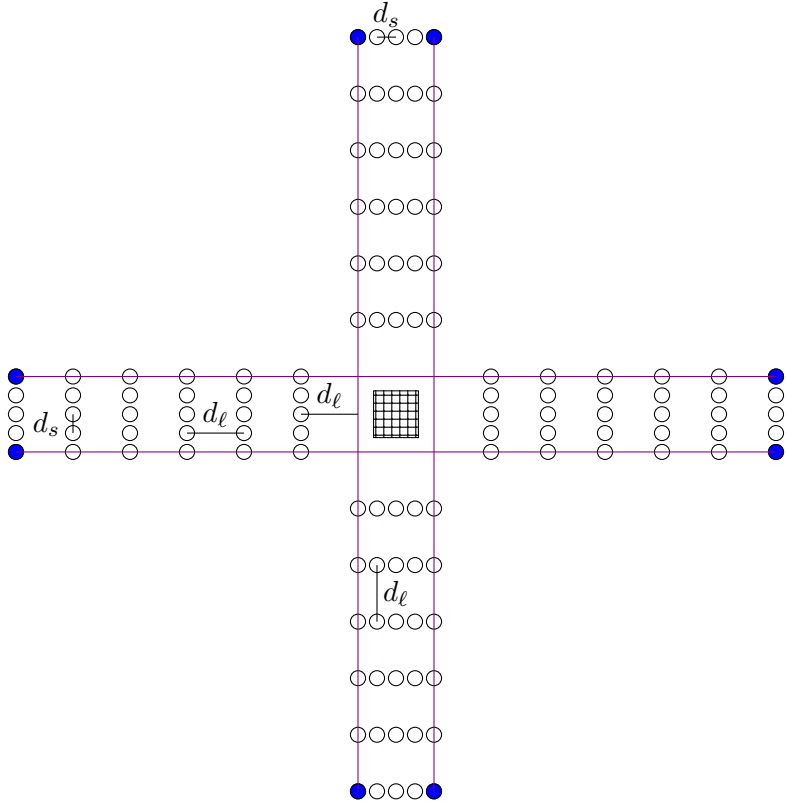

    Next, we construct the rest of the instance in the middle gap.
    We start with~$\ell$~bundles of $\nu$ horizontal lines each. That is, for each~$i \in [\ell]$ and each~$j\in [\nu]$, we define a horizontal line~$h_i^j$.
    The height of the~$\nicefrac{\nu}{2}$\textsuperscript{th} line in the~$i$\textsuperscript{th} bundle is equal to the height of points in~$F$ in the~$i+1$\textsuperscript{st} row of~$G_h$.
    The distance between~$h_i^j$ and~$h_{i}^{j+1}$ is~$3$ where~$h_i^j$ is the higher line.
    Next, we define vertical lines~$v_i^j$ and~$s_i^j$ for each~$i \in [\ell]$ and each~$j\in [\nu]$ as follows.
    The $x$-position of line~$v_i^{\nicefrac{\nu}{2}} $ is equal to the $x$-positions of points in~$F$ in the~$i+1$\textsuperscript{st} column of~$G_v$.
    The distance between~$v_i^j$ and~$v_{i}^{j+1}$ is~$10n^2$ where~$v_i^j$ is to the left.
    Line~$s_{i}^j$ is one to the left of~$v_i^j$.
    An illustration of these lines (and the input points we will place on these lines next) is given in \Cref{fig:internal}.
    \begin{figure}
        \centering
        \begin{tikzpicture}
            \def\x{3}
            \def\y{5}
            \draw[dashed] (-\y,-1.5) to (\y,-1.5);
            \draw[dashed] (-\y,-1) to (\y,-1);
            \draw[dashed] (-\y,-.5) to (\y,-.5);
            
            \draw[dashed] (-\y,.5) to (\y,.5);
            \draw[dashed] (-\y,1) to (\y,1);
            \draw[dashed] (-\y,1.5) to (\y,1.5);

            \draw[dashed] (-3.3,\x) to (-3.3,-\x);
            \draw[dashed] (-3,\x) to (-3,-\x);
            \draw[dashed] (-2.3,\x) to (-2.3,-\x);
            \draw[dashed] (-2,\x) to (-2,-\x);
            \draw[dashed] (-1.3,\x) to (-1.3,-\x);
            \draw[dashed] (-1,\x) to (-1,-\x);
            
            \draw[dashed] (3.3,\x) to (3.3,-\x);
            \draw[dashed] (3,\x) to (3,-\x);
            \draw[dashed] (2.3,\x) to (2.3,-\x);
            \draw[dashed] (2,\x) to (2,-\x);
            \draw[dashed] (1.3,\x) to (1.3,-\x);
            \draw[dashed] (1,\x) to (1,-\x);
            
            \draw[violet] (-\y,2.5) to (\y,2.5);
            \draw[violet] (-\y,-2.5) to (\y,-2.5);

            \draw[violet] (-4.5,-\x) to (-4.5,\x);
            \draw[violet] (4.5,-\x) to (4.5,\x);

            \draw[<->] (-3.3,2.8) to (1,2.8);
            \node at(-1.6,3.2) {$\dsmall$};

            \draw[<->] (1.3,2.8) to (2,2.8);
            \node at(1.65,3.2) {$10n^2-1$};

            \draw[<->] (3,2.8) to (3.3,2.8);
            \node at(3.15,3.2) {$1$};
            
            \draw[<->] (-3.5,1.5) to (-3.5,-.5);
            \node at(-3.8,.3) {$\dsmall$};

            \draw[<->] (-3.5,-1) to (-3.5,-1.5);
            \node at(-3.8,-1.25) {$3$};
            
            \node[point,fill=black] at(-3,1.5) {};
            \node[point,fill=black] at(-2.3,1.5) {};
            \node[point,fill=black] at(-1.3,1.5) {};
            \node[point,fill=black] at(-3.3,1) {};
            \node[point,fill=black] at(-2,1) {};
            \node[point,fill=black] at(-1.3,1) {};
            \node[point,fill=black] at(-3.3,.5) {};
            \node[point,fill=black] at(-2.3,.5) {};
            \node[point,fill=black] at(-1,.5) {};
            
            \node[point,fill=black] at(1,1.5) {};
            \node[point,fill=black] at(2.3,1.5) {};
            \node[point,fill=black] at(3.3,1.5) {};
            \node[point,fill=black] at(1.3,1) {};
            \node[point,fill=black] at(2.3,1) {};
            \node[point,fill=black] at(3,1) {};
            \node[point,fill=black] at(1.3,.5) {};
            \node[point,fill=black] at(2,.5) {};
            \node[point,fill=black] at(3.3,.5) {};
            
            \node[point,fill=black] at(-3.3,-.5) {};
            \node[point,fill=black] at(-2,-.5) {};
            \node[point,fill=black] at(-1,-.5) {};
            \node[point,fill=black] at(-3,-1) {};
            \node[point,fill=black] at(-2,-1) {};
            \node[point,fill=black] at(-1.3,-1) {};
            \node[point,fill=black] at(-3,-1.5) {};
            \node[point,fill=black] at(-2.3,-1.5) {};
            \node[point,fill=black] at(-1,-1.5) {};
            
            \node[point,fill=black] at(1.3,-.5) {};
            \node[point,fill=black] at(2,-.5) {};
            \node[point,fill=black] at(3,-.5) {};
            \node[point,fill=black] at(1,-1) {};
            \node[point,fill=black] at(2.3,-1) {};
            \node[point,fill=black] at(3,-1) {};
            \node[point,fill=black] at(1,-1.5) {};
            \node[point,fill=black] at(2,-1.5) {};
            \node[point,fill=black] at(3.3,-1.5) {};

            \foreach \i in {1,2,3}{
                \node[point,fill=blue] at(-4.7,\i*.5) {};
                \node[point,fill=blue] at(-4.3,\i*.5) {};
                \node[point,fill=blue] at(4.3,\i*.5) {};
                \node[point,fill=blue] at(4.7,\i*.5) {};
                \node[point,fill=blue] at(-4.7,-\i*.5) {};
                \node[point,fill=blue] at(-4.3,-\i*.5) {};
                \node[point,fill=blue] at(4.3,-\i*.5) {};
                \node[point,fill=blue] at(4.7,-\i*.5) {};

                \node[point,fill=blue] at(-\i-.3,2.7) {};
                \node[point,fill=blue] at(\i,2.7) {};
                \node[point,fill=blue] at(-\i-.3,2.3) {};
                \node[point,fill=blue] at(\i,2.3) {};
                \node[point,fill=blue] at(-\i-.3,-2.7) {};
                \node[point,fill=blue] at(\i,-2.7) {};
                \node[point,fill=blue] at(-\i-.3,-2.3) {};
                \node[point,fill=blue] at(\i,-2.3) {};
            }

            \draw[<->] (-4.7,1.7) to (-4.5,1.7);
            \draw[<->] (-4.3,1.7) to (-4.5,1.7);
            \node at(-4.6,2) {$1$};
            \node at(-4.4,2) {$1$};
            \node at(-5.3,1.5) {$h_1^1$};
            \node at(-5.3,-1) {$h_2^2$};
            \node at(-3,-3.3) {$v_1^1$};
            \node at(2.3,-3.3) {$v_2^2$};
            \node at(-3.3,-3.3) {$s_1^1$};
            \node at(1,-3.3) {$s_2^1$};
        \end{tikzpicture}
        \caption{An illustration of the construction of points in~$S \setminus F$ for an instance with $\ell = 2$ and~$\nu = 3$ vertices of each color. Black dots represent~$p$ points at the same coordinate, blue dots represent roughly $W$ points at the same coordinate, and the violet lines represent the fixed lines. Figure not drawn to scale since~$\dsmall$ is too big.}
        \label{fig:internal}
    \end{figure}
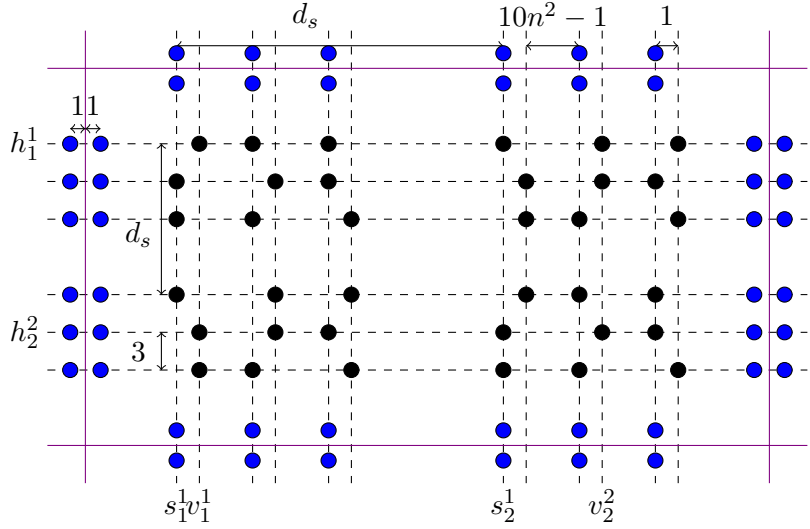
    For each pair~$w_j^i,w_{j'}^{i'}$ of vertices in the input graph~$G$ with~$i \neq i'$, we place~$p$ points on the intersection of~$h_i^j$ and~$s_{i'}^{j'}$ if~$\{w_j^i,w_{j'}^{i'}\} \in E$ and on the intersection of~$h_i^j$ and~$v_{i'}^{j'}$ otherwise.
    Note that we also place points symmetrically on the intersection of~$h_{i'}^{j'}$ and~$s_i^j$ if the two vertices are adjacent and on the intersection of~$h_{i'}^{j'}$ and~$v_i^j$ if they are not.
    For~$i = i'$, we place~$p$ points on the intersection of~$h_i^{j}$ and~$v_{i'}^{j'}$ if~$j \neq j'$ and on the intersection of~$h_i^j$ and~$s_{i'}^{j'}$ if~$j = j'$.
    Let~$X$ be the set of points placed in the previous step.
    Note that we place exactly~$n p = n^{11}$ points of~$X$ on each line~$h^j_i$, exactly~$(q+\nu-1) p < n^{11}$ points on each line~$s_i^j$, and exactly~$(n - q - \nu + 1)p < n^{11}$ points on each line~$v_i^j$.
    Finally, for each~$i \in [\ell]$ and each~$j \in [\nu]$, we place~$\frac{W}{4}$ points on~$s_i^j$ one above and below each of the two horizontal fixed lines (that is, we place~$W$ points on each line~$s_i^j$).
    Let~$Z_v$ be the set of all of these points.
    Note that~$|Z_v| = nW$.
    We also place~$\frac{W + \varphi(j)}{4} \geq \frac{W}{4}$ points on each line~$h_i^j$ one to the left and one to the right of each of the two vertical fixed lines.
    Let~$Z_h$ be the set of all these points.
    Note that~$|Z_h| = \ell \big( \sum_{j=1}^\nu W + \varphi(j)\big)$.
    This concludes the construction.

    \medskip

Since the construction can be computed in polynomial (in $n$) time as there are a polynomial number of points, all coordinates are integers and the largest distance between points is in~$O(\ell \dlarge) \subseteq \poly(n)$ and since~$k \in O(\ell)$, it only remains to show correctness.

\iflong
To this end, first assume that the original instance of \textsc{Regular Multicolored Independent Set} is a yes-instance, that is, $V$ contains a multicolored independent set~$I = \{w_{j_1}^1,w_{j_2}^2,\ldots,w_{j_\ell}^\ell\}$.
We take the four fixed lines and the lines~$h_i^{j_i}$ and~$s_i^{j_i}$ for each~$i \in [\ell]$ into the solution. Note that we picked exactly~$2\ell+4 = k$ lines.
Moreover, each of the 8 positions where~$\many$ points are placed on the same position are hit exactly by the four fixed lines. Each of the other~$2(k+2)(\frac{k}{2}-2) \leq \ell^3$ points in~$F$ have distance at most~$10n^2\frac{\nu}{2}$ from one of the picked lines and hence all these points contribute at most~$\ell^3 (5n^2\nu)^2 \leq 25\ell n^6 \leq n^{7}$ to the budget (the sum of squares of distances from each point to the closest line) in total.
Next, note that all points in~$Z_h \cup Z_v$ have distance one from one of the four fixed lines.
Moreover, all points on the solution lines~$h_i^{j_i}$ and~$s_i^{j_i}$ have distance~0 to the solution.
Hence, all points in~$Z_v$ contribute~$(n-\ell)W$ to the budget.
Similarly, the points in~$Z_h$ contribute~$|Z_h| - \sum_{i \in [\ell]} (W + \varphi(j_i))$ to the budget. 
Finally, we analyze the points in~$X$.
Note that for each~$i \in [\ell]$, all points on lines~$h_i^j$ for any~$j \in [\nu]$ are closest to~$h_i^{j_i}$ except for those points that are on lines~$s_{i'}^{j_{i'}}$ or~$v_{i'}^{j_{i'}}$ for some~$i' \in [\ell]$. Moreover, all of the points on~$h_i^j$ have distance~$3 (j - j_i)$ from~$h_i^{j_i}$ if~$j_i \leq j$ and distance~$3 (j_i - i)$ if~$j_i \geq j$. Since there are $p \ell (\nu-1)$ points on each of these lines that are closest to~$h_i^{j_i}$, they contribute exactly~$p(\nu-1)\ell(\sum_{a=1}^{j_i}(3(j_i-a))^2 + \sum_{b=j_i}^{\nu} (3(\nu - b))^2) = \varphi(j_i)$ to the budget (which summed up over all lines~$h_i^j$ is precisely the difference between~$|Z_h|-W\ell$ and what the points in~$Z_h$ contribute to the budget). 

Finally for each $i \in [\ell]$, all points on~$s_i^{j_i}$ contribute~0 to the budget and all points on~$v_i^{j_i}$ contribute one each except for the points on~$h_{i}^{j_i}$.
Since we chose the solution lines corresponding to an independent set, all these points exist and hence all points on~$v_i^{j_i}$ contribute~$p (n - q - \nu + 1- \ell)$ to the budget.

Overall, the sum of squares of distances of all points to the closest line in the solution is at most 

$$n^7 + (n-\ell)W + \ell \big( \sum_{j=1}^\nu W + \varphi(j)\big) - \ell W + \ell p(n-\nu + 1-q-\ell) = B,$$ concluding the proof for the forward direction.

Now assume that the constructed instance is a yes-instance.
First, note that each point has distance at most~$\sqrt{B} \leq d_s$ from one of the solution lines (as otherwise this single point contributes more than~$B$ to the budget).
Since~$F$ contains points on~$k^2+2k$ different positions and the solution contains~$k$ lines, on average, each line in the solution has distance at most~$d_s$ from~$(k+2)$ positions with points in~$F$.
Assume such a line~$L$ is not close to all~$(k+2)$ positions in one row of $G_h$ or one column of~$G_v$.
    Moreover, assume without loss of generality that~$L$ is close to at least~$\frac{k+2}{2}$ positions in~$G_h$ (it is close to at least that many points in one of the two grids~$G_h$ and~$G_v$ and the other case is symmetric).
    Since~$\frac{k+2}{2} > \frac{k}{2}$, it holds that~$L$ is close to positions in at least two different columns of~$G_h$.
    As the minimal distance between two columns is~$d_\ell$ and the maximum distance between two rows is at most~$(k+2) d_s$, it holds that the inclination of~$L$ is between~$-\frac{(k+2) d_s + 2\sqrt{B}}{d_\ell-2\sqrt{B}} \geq -\frac{1}{n^{48}}$ and~$\frac{(k+2) d_s + 2\sqrt{B}}{d_\ell-2\sqrt{B}} \leq \frac{1}{n^{48}}$.
    Note that~$L$ cannot be close to any point in~$F$ on~$G_v$ as the vertical distance between such a point and any point in~$G_h$ is at least~$d_\ell$ and the inclination of~$L$ over the width of~$G_h$~($(k+2) \cdot d_\ell$) only allows for a change of at most~$(k+2)^2 d_s$ which is less than~$d_\ell - 2\sqrt{B}$.
    Hence~$L$ is close to~$k+2$ points in~$F$ on~$G_h$.
    Moreover, $L$ is close to exactly one point in each column of~$G_h$.
    Finally, if~$L$ is close to points in different rows of~$G_h$, then the inclination of~$L$ is larger than~$y =\frac{\dsmall-2\sqrt{B}}{(k+3)\dlarge}$ (or smaller than~$-y$).
    Hence, $L$ cannot be close to two points in the same row (as this would allow for a maximum absolute incline of~$\frac{2\sqrt{B}}{\dlarge} < y$).
    As~$G_h$ contains only~$\frac{k}{2} < k+2$ rows, this contradicts that~$L$ is close to~$k+2$ points in~$G_h$.
    This implies that each line in the solution is close to at most~$k+2$ to points in~$F$.
    By pigeonhole principle, each line in the solution is thus close to exactly~$k+2$ points in~$F$ and hence each line in the solution is close to one (distinct) row in~$G_h$ or one (distinct) column in~$G_v$.
    
    We next analyze the fixed lines.
    As proven above, each line stays within a band of width~$\sqrt{B}$ around one row of~$G_h$ or one column of~$G_v$.
    Note that in the top- and bottom-most rows of~$G_h$ there are $\many$ points at the very left and very right.
    Hence, the solution line has to stay within distance at most~$\sqrt{\frac{B}{\many}} \leq \frac{1}{n^{34}}$ from the fixed line.

    Next, consider a solution line~$L$ that is close to the~$i$\textsuperscript{th} row of~$G_h$ for some~${i \in [\frac{k}{2}] \setminus \{1,\frac{k}{2}\}}$.
    If~$L$ is not the closest line to four positions where points of~$Z_h$ are placed, then we can replace~$L$ by a line~$L'$ that goes through four such positions exactly and this will improve the solution as we reduce the solution by at least~$\frac{W}{4}$.
    Since~$X$ contains less than~$n^2p = n^{12}$ points and the distance from the new line~$L'$ to each point in~$X$ on some line~$h^{j'}_i$ (points close to the~$i$\textsuperscript{th} row of~$G_h$) is at most~$3\nu$, we increase the budget by at most~$9\nu^2n^{12} < \frac{W}{n}$ for the points in~$X$.
    Moreover, $L'$ also has distance at most~$3\nu$ from each point in~$F$ that is close to~$L$.
    Thus, replacing~$L$ by~$L'$ indeed improves the solution.
    Assuming that~$L$ is closest to four positions with more than~$\frac{W}{4}$ points of~$Z_h$ each, assume towards a contradiction that~$L$ is closest to some position on two different lines~$h_i^j,h_i^{j'}$ for some~$j \neq j' \in [\nu]$.
    Assume without loss of generality that~$j < j'$ and that~$L$ has distance at most~$1$ to~$\frac{W}{4}$ points of~$h_i^j$ left of the position where it has distance at most~$1$ from~$\frac{W}{4}$ points on~$h_i^{j'}$.
    Then, $L$ has an incline of at least~$\frac{1}{\ell d_s} > \frac{1}{n^{41}} > \frac{1}{n^{47}}$, a contradiction to the fact that~$L$ has distance at most~$\sqrt{B}$ from all points in the~$i+1$\textsuperscript{st} row of~$G_h$.
    Thus, $L$ has distance at most~$1$ from some line~$h_i^j$ throughout the entire middle part of the construction (the area where~$G_h$ and~$G_v$ overlap).
    Analogously, each solution line~$L$ that is close to some column of~$G_v$ (except the fixed lines) has to reduce the budget by~$W$ and hence stays within distance one of some line~$s_i^j$ in the middle part.

    Let us now give a lower bound on the budget used by all points in a solution assuming that all points in~$X$ are closest to a ``horizontal'' line (one that has distance at most one from a line~$h_i^j$ for some~$i \in [\ell]$ and~$j \in [\nu]$).
    We will assume that all points in~$F$ do not contribute anything (they just enforce that a solution has the structure as analyzed above).
    Let us start with the case where all lines of the solution are~$h_i^j$ or~$s_{i}^j$ for some~$i \in [\ell]$ and some~$j \in [\nu]$.
    Then, all but~$\ell W$ points in~$Z_v$ (these are~$nW - \ell W$ many) have distance~$1$ to a solution line and the remaining~$\ell W$ points have distance~$0$ to a line in the solution.
    For each~${i \in [\ell]}$, let~$j_i$ be the line such that the solution contains line~$h_i^j$.
    All points in~$X$ that are closest to~$h_i^{j_i}$ have distance~$pn(\sum_{a=1}^{j_i}(3(j_i-a))^2 + \sum_{b=j_i}^{\nu} (3(\nu - b))^2) = \varphi'(j_i)$ from~$h_i^{j_i}$ combined.
    The points in~$Z_h$ which lie on~$h_i^{j_i}$ contribute nothing and all points in~$Z_h$ that lie on some line~$h_i^j$ for some~$j \neq j_i$ contribute one each.
    Hence, all points in~$Z_v \cup X \cup Z_h$ contribute~$(n-\ell) W + |Z_h| - \ell W + \sum_{i=1}^{\ell}(\varphi'(j_i) - \varphi(j_i))$ combined.
    Hence, even this lower bound (ignoring the fact that some points in~$X$ can be closest to a vertical line) is equal to
    \begin{align*}
    &B - n^7 - \ell p (n - \nu + 1 - \ell -q) + \sum_{i=1}^{\ell} (\varphi'(j_i) - \varphi(j_i))\\ =\ &B - n^7 - \ell p (n - \nu + 1 - \ell -q) + \sum_{i=1}^{\ell}p \ell \vartheta(j_i) > B.
    \end{align*} 
    Here, the last inequality follows from~${\vartheta(i) = (\sum_{a=1}^{i}(3(i-a))^2 + \sum_{b=i}^{\nu} (3(\nu - b))^2 > \nu^2 > n}$ for all~$i \in [\nu]$ and~$\ell p > n^7$.
    Next, note that for each pair of lines~$h_i^{j_i}$ and~$s_{i'}^{j'}$ in the solution, it holds that exactly~$p(\nu-1)$ of the points which we previously assumed to be closest to~$h_i^{j_i}$ are actually closer to~$s_{i'}^{j'}$.
    These points have a combined square of distance  from~$h_i^{j_i}$ of
    \[p (\sum_{a=1}^{j_i-1}(3(j_i-a))^2 + \sum_{b=j_i+1}^{\nu} (3(\nu - b))^2) = p \vartheta(j_i).\]
    Summed over all combinations of lines, this reduces the budget by~$\sum_{i=1}^{\ell}p\ell \vartheta(j_i)$.
    Next, if we chose lines~$h_i^{j_i}$ and~$s_{i'}^{j'}$ such that~$j' = j_i$ whenever~$i = i'$ and~$\{w^i_{j_i},w^{i'}_{j'}\} \notin E$ whenever~$i \neq i'$, then $v_{i'}^{j_i}$ contains exactly~$p(n-\nu+1-\ell - q)$ points that are not on some solution line~$h^{i'}_{j_{i'}}$.
    The updated lower bound is hence
    \[B - n^7 - \ell p (n-\nu+1-\ell-q) + \sum_{i=1}^{\ell} p \ell \vartheta(j_i) - \sum_{i=1}^{\ell}p\ell \vartheta(j_i) + \ell p (n-\nu+1-\ell-q) = B - n^7.\]
    If we pick lines differently, then this increases the budget by at least~$p > n^7$, which proves that if we restrict our solution to contain lines~$h_i^j$ and~$s_i^j$ exactly, then it must correspond to a colorful independent set.
    Moreover, the gap for no-instances (with the above restrictions) is~$p - n^7$.
    
    To conclude the proof, note that we already assumed that points in~$F$ do not contribute to the budget and we cannot reduce the contribution of vertices in~$Z_h \cup Z_v$ by adjusting the solution lines slightly since each line will be closest to four positions with~$\frac{W}{4}$ points each and in the considered solution, all these points are hit exactly (so they do not contribute to the budget).
    As analyzed above, each line in the solution is within distance~$1$ from~$h_i^{j}$ or~$s_{i}^j$ for some~$i\in [\ell]$ and some~$j \in [\nu]$.
    For a solution line~$L$ that is within distance one from~$h_i^{j}$, let~$\varepsilon_1\geq 0$ be the distance between~$L$ and~$h_i^j$ at the left fixed line and let~$\varepsilon_2\geq 0$ be the distance at the right fixed line.
    Note that the distance between each point in~$X$ and~$L$ and~$h_i^j$, respectively, differs by at most~$\varepsilon = \max(\varepsilon_1,\varepsilon_2)$.
    As the distance between~$h_i^j$ and a point in~$X$ which is assigned to~$L$ in the solution is at most~$3n\nu$, the contribution of the point if we replace~$h_i^j$ by~$L$ decreases by at most~$(3n\nu)^2 - (3n\nu-\varepsilon)^2 \leq 6n\nu\varepsilon$.
    However, the contribution of at least~$\frac{W}{4}$ points in~$Z_h$ increases by at least~$\varepsilon^2$ each.
    As~$|X| = n^{12}$, for this exchange to decrease the budget, it has to hold that $6n^{13}\nu\varepsilon > \frac{W}{4}\varepsilon^2$, which implies that~$\varepsilon \leq \frac{1}{n^{15}}$.
    Hence, the overall reduction in budget we can get from all vertices in~$X$ combined if we slightly shift or tilt all horizontal lines in the solution is at most~$|X|((3n\nu)^2-(3n\nu-\nicefrac{1}{n^{15}})^2) \leq n^{12}\frac{6\nu}{n^{15}} < 1$.
    The argument for lines that are close to a column of~$G_v$ is basically the same.
    Note that since~$p - n^7 > 16$, even with these additional budget savings, we cannot afford to pick lines into a solution that do not correspond to an independent set as a single additional position with~$p$ points where the distance is more than~$\frac{1}{2}$ increases the budget to above~$B+2$.
    This concludes the proof.
\fi
\end{proof}


\section{\classNP-hardness of covering points with two subspaces}

\newcommand{\linf}{\ell_{\infty}}
\newcommand{\ball}{\mathcal{B}}
\newcommand{\ourproblem}{\textsc{ourproblem}\xspace}
\newcommand{\dsfull}{\textsc{Dominating Set}\xspace}
\newcommand{\ds}{\textsc{DS}\xspace}
\newcommand{\reals}{\mathbb{R}}
\newcommand{\nz}{\textsc{NZ}}
\newcommand{\z}{\textsc{Z}}
\newcommand{\bigoh}{\mathcal{O}}

\newcommand{\ecnpd}{\textsc{ECNP-d}\xspace}
\newcommand{\ecnp}{\textsc{ECNP}\xspace}

In this section, we show that \problemHLC is \classNP-hard even when $k=2$. 
To this end, we present a reduction from \textsc{Equal-Cardinality Number Partition} (\textsc{ECNP}, in short).
In \ecnp, we are given a set $W = \{ w_1, \dots, w_{2n} \}$ of $2n$ positive integers that sum to $2t$, and the question is whether $W$ can be partitioned into two disjoint subsets $R_1$ and $R_2$ such that $|R_1| = |R_2| = n$, and $\sum_{w \in R_1} w = \sum_{w \in R_2} w = t$.
\ecnp is known to be \classNP-hard \cite{GareyJ79}.

Let \textsc{Equal-Cardinality Number Partition with distinct numbers} (\textsc{ECNP-d}, in short) denote the restriction of the problem where the numbers in the input are distinct. 
It is a folklore knowledge that \ecnpd is \classNP-hard (the result can be shown by disturbing the integers in the input of \ecnp). 
The basic idea of the following reduction is to capture the $n$ distinct numbers that add up to $t$ via a hyperplane.

\objzerohard*
\begin{proof}
    Let $(W = \{ w_1, \dots, w_{2n} \}, t)$ denote an instance of \ecnpd.
    We construct a point set $P$ in $\mathbb{R}^{n-1}$ as follows.
    For each $w_i \in W$, we add the point:
    \[
    \mathbf{p_i} = (v_i, v_i^2, \dots, v_i^{n-2}, v_i^n),
    \]
    where $v_i = w_i - t/n$.
    Let $V = \{ v_1, \dots, v_{2n}\}$ denote the set of shifted values.
    Observe that an equal-cardinality partition in $W$ yields an equal-cardinality partition in $V$ where each part has sum $0$.
    The total number of points in $P$ is $2n$. 
    We claim that $P$ can be covered by exactly $2$ hyperplanes if and only if $V$ admits a valid equal-cardinality partition with sum $0$ (which is true if and only if $W$ admits a valid equal-cardinality partition).
    
    Suppose there exists a valid partition $(R_1, R_2)$ of $V$ where $|R_1| = |R_2| = n$ and the elements of each subset sum to exactly $0$.
    We define two polynomials of degree $n$:
    \[
    Q_1(y) = \prod_{v \in R_1} (y - v) = y^n - \left(\sum_{v \in R_1} v\right)y^{n-1} + \dots + (-1)^n \prod_{v \in R_1} v 
    \]
    and 
    \[
    Q_2(y) = \prod_{v \in R_2} (y - v) = y^n - \left(\sum_{v \in R_2} v\right)y^{n-1} + \dots + (-1)^n \prod_{v \in R_2} v.
    \]
    Let the coefficients of $Q_1(y)$ be $a_n, a_{n-1}, \dots, a_0$ such that $Q_1(y) = \sum_{j=0}^n a_j y^j$, and let the coefficients of $Q_2(y)$ be $b_n, b_{n-1}, \dots, b_0$. 
    Note that $a_{n-1} = b_{n-1} = 0$ since the parts sum to $0$.
    We define the hyperplane $H_1$ by the equation $a_0 + a_1 x_1 + \dots + a_{n-2} x_{n-2} + a_n x_{n-1} = 0$.
    Analogously, we define $H_2$ using the coefficients of $Q_2(y)$.
    By construction, evaluating $H_1$ on point $\mathbf{p_i}$ perfectly mirrors evaluating $Q_1(v_i)$. 
    Since the roots of $Q_1$ are exactly the elements of $R_1$, $H_1$ evaluates to $0$ and safely covers $\mathbf{p_i}$ for all $v_i \in R_1$. 
    Analogously, $H_2$ covers $\mathbf{p_i}$ for all $v_i \in R_2$. 
    Thus, all $2n$ points of the form $\mathbf{p_i}$ are covered.
    The point set $P$ is entirely covered by $H_1$ and $H_2$.

    Suppose $P$ can be covered by exactly two hyperplanes (in $\mathbb{R}^{n-1}$), $H_1$ and $H_2$.
    Let the equation for $H_1$ be $a_0 + a_1 x_1 + \dots + a_{n-2} x_{n-2} + a_n x_{n-1} = 0$ and $H_2$ be $b_0 + b_1 x_1 + \dots + b_{n-2} x_{n-2} + b_n x_{n-1} = 0$.
    We deliberately use $a_n$ and $b_n$ for the sake of presentation.
    For any point $\mathbf{p_i}$ covered by $H_1$, we have $a_0 + a_1 v_i + a_2 v_i^2 + \dots + a_{n-2} v_i^{n-2} + a_n v_i^n = 0$. 
    Thus, $v_i$ is a root of the polynomial 
    \[
    Q_1(y) = a_0 + a_1 y + a_2 y^2 + \dots + a_{n-2} y^{n-2} + 0 \cdot y^{n-1} + a_n y^n.
    \]
    Since $H_1$ is a valid hyperplane, 
    $(a_0,\ldots,a_{n-2},a_n)$
    cannot be the zero vector; thus $Q_1(y)$ is not the zero polynomial. 
    Since $Q_1(y)$ has degree at most $n$, and the numbers $v_i$ are distinct, it can have at most $n$ distinct roots. 
    Consequently, $H_1$ can cover at most $n$ points of the form $\mathbf{p_i}$. 
    By symmetric arguments, $H_2$ can also cover at most $n$ points of the form $\mathbf{p_i}$.
    We define the polynomial $Q_2(y)$ analogously.
    Since there are $2n$ points of the form $\mathbf{p_i}$ in total, $H_1$ is forced to cover exactly $n$ such points (corresponding to subset $R_1$), and $H_2$ must cover exactly the remaining $n$ such points (subset $R_2$). 
    This guarantees that $a_n \neq 0$ and $b_n \neq 0$ (otherwise the degree drops and they could not cover $n$ points).
    The sum of the roots of $Q_1(y)$ is given by: $-a_{n-1}/a_n = 0$.
    Thus, we infer that $\sum_{v \in R_1} v = 0$.
    Consequently, the corresponding original elements $w_i = v_i + t/n$ sum to $\sum_{v_i \in R_1} (v_i + t/n) = 0 + n(t/n) = t$, and we obtain the partition that forms a valid solution to the \ecnpd instance.
\end{proof}

\section{\xp algorithm for \projcl}

In this section, we prove \Cref{theorem:algoforprojcl}. Its proof uses a fundamental result from algebraic geometry. We need a few definitions to state this result. 
%
%
Let $\bbR[X_1,X_2,\ldots,X_d]$ be the ring of polynomials in variables $X_1,\ldots,X_d$ with coefficients in $\bbR$. Let ${V=\{x \in \bbR^d\mid Q(x_1,\ldots, x_d)=0\}}$ denote an algebraic set in $\bbR^d$ defined by~$Q \in \bbR[X_1,X_2,\ldots,X_d]$. The sign condition for a set of polynomials is defined as follows. Let~$\calP=\{P_1,\ldots P_s\}\subseteq \bbR[X_1,X_2,\ldots,X_d]$ be a subset of $s$ polynomials. The sign condition for~$\calP$ is specified by a sign vector $\sigma \in \{-1,0,+1\}^s$ and the sign condition is non-empty over~$V$ with respect to $\calP$ if there is $x \in V$ such that~$
\sigma=(\sgn(P_1(x)),\ldots, \sgn(P_s(x))),
$
where $\sgn(y)$ is the sign function defined as 
\[
\sgn(y)=
    \begin{cases}
       1,\text{ if }y>0,\\
       0,\text{ if }y=0, \text{ and}\\
       -1,\text{ if }y<0,
    \end{cases}
\]
for $y \in \bbR$.
\begin{figure}[t!]
    \centering
    \includegraphics[width=0.9\linewidth]{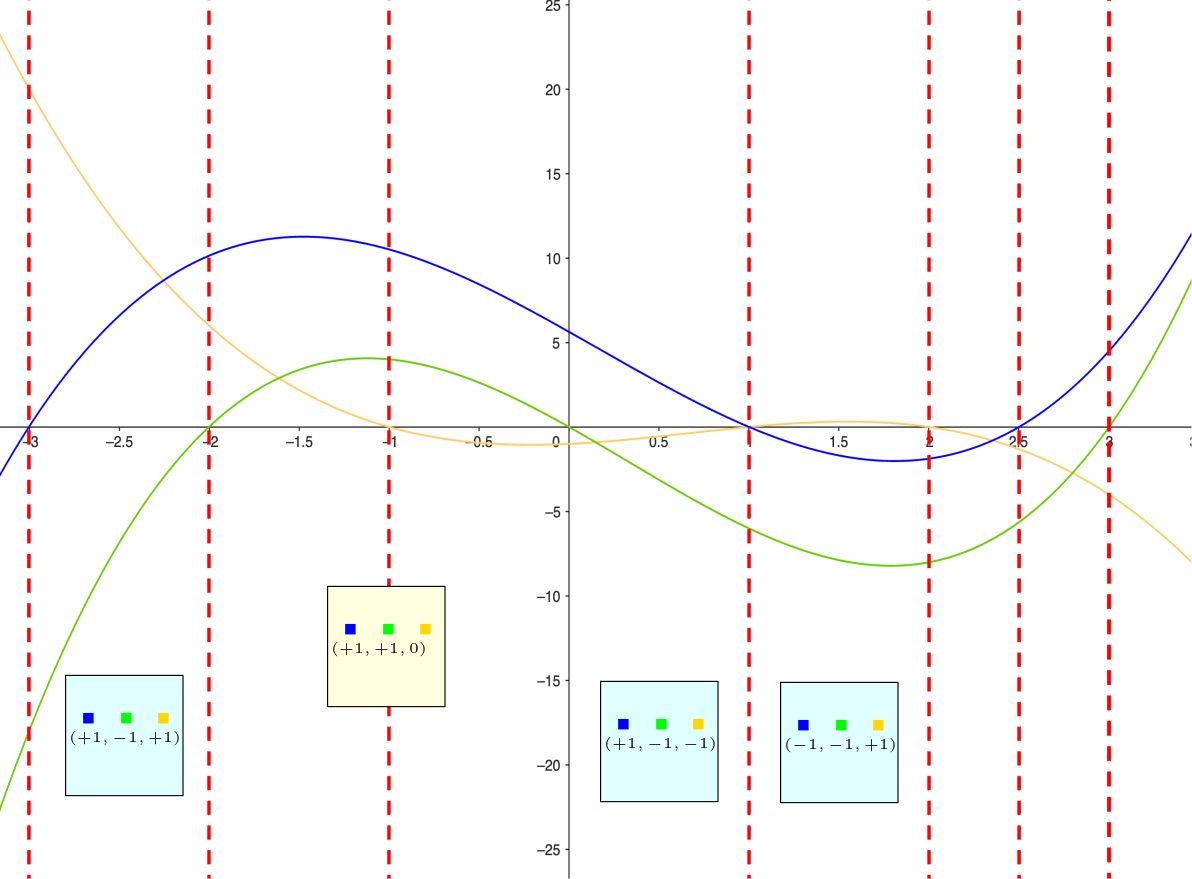}
    \caption{ A family of 3 polynomials on a single variable~$x$. The red dotted line 
     partitions the $x$-axis into cells. The sign condition of some cells are given. The yellow box shows the sign condition of the cell consisting only of the point~$-1$. The cells are $(-\infty, -3), \{-3\}, (-3, -2)$ and so on.}
    \label{fig:signcond}
\end{figure}
The \textit{realization space} of $\sigma \in \{-1,0,+1\}^s$ over $V$ is the set 
\[
R(\sigma)= \big{\{}x\mid x\in V, \sigma = (\sgn(P_1(x)),\ldots,\sgn(P_s(x))) \big{\}}.
\]
If $R(\sigma)$ is not empty, then each of its non-empty semi-algebraically connected (which is equivalent to just connected on semi-algebraic sets as proven in \cite[Theorem 5.22]{10.5555/1197095}) components is a \textit{cell} of $\calP$ over~$V$.
\Cref{fig:signcond} shows an example of a family of polynomials on a single variable partitioning the $x$-axis into cells.
For an algebraic set $W$, its real dimension is the maximal integer $d'$ such that there is a homeomorphism of $[0,1]^{d'}$ in $W$. Naturally if $W \subset \bbR^d$, then $d'<d$. The following theorem by Basu et al.~\cite{10.5555/1197095} gives an algorithm to compute a point in each cell of~$\calP$ over~$V$.

\begin{proposition}[{\cite[Theorem 13.22]{10.5555/1197095}}]\label{proposition:alggeom} Let $V$ be an algebraic set in $\bbR^d$ of real dimension $d'$ defined by $Q(X_1,\ldots, X_d)=0$, where 
$Q$ is a polynomial in $\bbR[X_1,\ldots, X_d]$ of degree at most~$b$, and let $\calP \subset \bbR[X_1,X_2,\ldots,X_d]$ be a finite set of $s$ polynomials with each $P\in \calP$ also of degree at most $b$. Let $D$ be a ring generated by the coefficients of $Q$ and the polynomials in $\calP$. There is an algorithm that takes as input $Q,d'$ and $\calP$ and computes a set of points meeting every non-empty cell of $V$ over $\mathcal{P}$. The algorithm uses at most $s^{d'}b^{\mathcal{O}(d)}$ arithmetic operations in $D$. 
\end{proposition}

We next present our algorithm.
We start by analyzing the special case where~$k=1$.
Here and in the more general case where~$k>1$, we will distinguish between the cases~$r=0$ and~$r \geq 1$.
Recall that for~$r=0$, the problem is equivalent to the $k$-means problem.
For~$k=1$, the solution is the centroid of the given points which can be computed in $\mathcal{O}(nd)$ time. The centroid \( C \) of a set of points \( \{x_i\}_{i=1}^n \) in \( d \)-dimensional space is the mean of all the points. It is calculated using the formula~\(C = \frac{1}{n} \sum_{i=1}^n x_i.\)
For $r>0$, we have the following proposition.
 
\begin{proposition}\label{prop:basecase}
 For $k=1$, \projcl{} is solvable in time $\Oh(d n^2+d^3)$.
\end{proposition}

The proof of \Cref{prop:basecase} follows from the well-known reduction that allows reducing the problem of finding the \emph{affine subspace} that minimizes the sum of squares of distances to a given point set to the problem of finding a \emph{subspace} with the same property for a similar point set.
Indeed, any~\( r \)-dimensional affine space that minimizes the sum of squared perpendicular distances to the data points must pass through the centroid of these points (see, e.g.,~\cite{Blum_Hopcroft_Kannan_2020}). Thus, by centering the data---i.e., by subtracting the centroid from each data point---we ensure that the best-fit affine subspace passes through the origin.
Consequently, the problem of finding the best-fit \( r \)-dimensional subspace is identical to the classic \emph{principal component analysis} (\pca). By the Eckart-Young theorem \cite{Eckart1936TheAO}, \pca can be solved efficiently via \emph{Singular Value Decomposition} (SVD), which can be implemented in~\(\Oh(d n^2 +d^3) \) time~\cite{trefethen1997numerical}.

Now, we will provide a brief outline of our algorithm for \projcl{} for~$k > 1$. Any optimal solution~$A_1,A_2,\ldots A_k$ will partition the set of points into $k$ clusters where each cluster corresponds to an affine subspace~$A_i$ which is the closest affine subspace for each point in that cluster. Thus, a trivial algorithm will be to guess the partition of the input points into~$k$~clusters and find the~$k$~affine subspaces by running the algorithm of \Cref{prop:basecase} for each cluster. Since the number of partitions of the $n$ points is at most $k^n$ we have reduced \projcl{} to $k^n$ instances of \pca. 
In order to reduce the search space significantly, we make use of~\Cref{proposition:alggeom}.

We parameterize each of the $k$ unknown affine subspaces of dimension $r$ by $r+1$ $d$-dimensional variables. For each input point $\bfp$  and two unknown affine subspaces~$A_i$ and~$A_j$, we can write as a polynomial inequality the condition that $\bfp$ is closer to~$A_i$ than to~$A_j$. With this idea, we encode the problem as the solution to systems of inequalities of~$nk^2$~polynomials of degree at most $4$ and real dimension $ dk(r+1)$. For affine subspaces corresponding to variables from the same cell of the corresponding algebraic set,  the set of closest points are the same.  This defines the partition of the points into clusters. For each such cluster we can find the best fit by making use of \Cref{prop:basecase}. Finally, 
by \Cref{proposition:alggeom}, the number of such partitions is $ n^{\mathcal{O}(dk(r+1))}$.

 
\xpalgo*
    

\begin{proof}
We distinguish between two cases: $r=0$ and~$r>0$.
In the first case, note that~$n^{\Oh(dk(r+1))} = n^{\Oh(dk)}$.
Moreover, the case is equivalent to $k$-\textsc{Means}.
We can therefore use the~$n^{\Oh(dk)}$-time algorithm of Inaba et al. \cite{inaba1994applications}. 

We next explain our algorithm for the case $r>0$.
We will map $k$-tuples of $r$-dimensional affine subspaces with points of a certain algebraic set. We can represent an $r$-dimensional affine subspace of $\mathbb{R}^d$ by a $d \times (r+1)$ matrix $\mathbf{V}$ where the first $r$ columns represent an orthonormal basis of a linear subspace of dimension $r$ and column~$(r+1)$ represents the offset point. 
Consider the matrix space $\mathbb{R}^{d \times k (r+1)}$ and a matrix~$\mathbf{V} \in \mathbb{R}^{d \times k(r+1)}$.
Let~$\mathbf{V}=\{v_{ij}\}$ for~$1\leq i \leq d$ and~$1\leq j \leq k(r+1)$. Consider the following polynomial conditions.
\begin{align*}
Q^{O}_{i,j,h}(\mathbf{V})&\coloneqq \sum_{\ell=1}^{d}v_{\ell i}v_{\ell j}=0 \text{ for each } h \in[0,k-1] \text{ and each } hr+1\leq i<j\leq (h+1)r\\
Q_i^{N}(\mathbf{V})&\coloneqq \left(\sum_{\ell=1}^{d}v_{\ell i}^2 \right)-1=0 \text{ for each }1 \leq i \leq kr
\end{align*}
The condition $Q^{O}_{i,j,h}(\mathbf{V})=0$ requires the columns $i,j$ of $\mathbf{V}$ where $hr+1\leq i<j\leq (h+1)r$ to be pairwise orthogonal and the condition $Q_i^{N}(\mathbf{V})=0$ requires the column $i$ to have length $1$.
We will next show how a matrix~$\mathbf{V}\in \mathbb{R}^{d \times k (r+1)}$ which satisfies all of the above conditions represents $k$-affine subspaces of dimension $r$. For $h \in [0,k-1]$ the columns $hr+1$ to $(h+1)r$ represent the orthonormal basis of the linear subspace corresponding to the $h$\textsuperscript{th} affine subspace of rank $r$ and column~$(kr+h+1)$ represents the corresponding offset point.
We combine all of the above equations into a single one by taking the sum of squares
\begin{align}
Q(\mathbf{V})= \mathop{\sum}_{\substack{ hr+1\leq i<j\leq (h+1)r \\  h \in[k] }}\left(Q^{O}_{i,j,h}(\mathbf{V})\right)^2+ \sum_{i=1}^{kr}\left(Q_i^{N}(\mathbf{V})\right)^2.\label{eq:0}
\end{align}


Note that $Q(\mathbf{V})=0$ if and only if~$Q^{O}_{i,j,h}(\mathbf{V}) = 0$ for all~$h,i,j$ and~$Q_i^{N}(\mathbf{V}) = 0$ for all~$i$.
Consider the algebraic set $W\subseteq \mathbb{R}^{d \times k(r+1)}$ of all matrices satisfying \cref{eq:0}.
We view each element of $W$ as a representation of $k$ affine subspaces of dimension $r$.
Let~$col_1,col_2,\ldots col_{d\times k(r+1)}$ be the columns of an element~$\mathbf{V} \in  W$. Then, the $h$\textsuperscript{th} affine subspace is represented by the linear subspace spanned by the columns $col_{hr+1}$ to $col_{(h+1)r+1}$ and the point corresponding to the column $col_{kr+h+1}$. Let $\mathbf{B_h}$ denote the submatrix formed by columns~$col_{hr+1}$ to $col_{(h+1)r+1}$.
The squared distance from a point~$\bfy \in \mathbb{R}^d$ to an $r$-dimensional affine subspace~$\bfp + \mathbf{B}x$ is given by $\norm{\bfy-\bfp-\mathbf{B}\mathbf{B}^T\bfy}_F$.

Now consider the family of polynomials $\calP=\{P_{i,j,f}\}_{1\leq i<j\leq k, 1\leq f\leq n}$ defined over $W$ by
\[
P_{i,j,f}(\mathbf{V})=\norm{\bfx_f-col_{kr+i+1}- \mathbf{B_i}\mathbf{B_i}^T\bfx_f}^2_F-\norm{\bfx_f-col_{kr+j+1}- \mathbf{B_j}\mathbf{B_j}^T\bfx_f}^2_F.
\]
The polynomial $P_{i,j,f}(\mathbf{V})>0$ if the $j$\textsuperscript{th} affine space represented by $\mathbf{V}$ is closer to $\bfx_f$ than the $i$\textsuperscript{th} affine space represented by $\mathbf{V}$ and vice-versa. For any particular $\mathbf{V} \in W$, the sign condition with respect to $\calP$ gives an ordering of the affine subspaces with respect to the Euclidean distance from each of the input points. Hence, we can partition the $n$ input points based on the nearest affine subspace.
Let $\calC$ be the partition of $W$ on cells over $\calP$. For each cell $C$, the sign condition with respect to $\calP$ is constant over $C$, which implies that the partitioning of the input points is the same over all points in $C$. Also, let $A_1,A_2,\ldots A_k$ be an optimal solution for \projcl{}. We can represent $A_i$ by $\mathbf{p_i}$ and~$\mathbf{B_i}$ where~$\mathbf{B_i} \in \mathbb{R}^{d\times r}$ and the columns of $\mathbf{B_i}$ are orthonormal. Here, $\mathbf{p_i}$ is the offset point and $\mathbf{B_i}$ is the $r$-dimensional linear subspace. Hence, there exists $\mathbf{V} \in W$ that represents~$A_1,A_2,\ldots A_k$. Thus if we run the classic \pca algorithm for each partition generated by every non-empty cell $C$, we will find an optimal solution. Thus our algorithm proceeds as follows:
\begin{itemize}
    \item Use the algorithm from~\Cref{proposition:alggeom} to obtain a point $\mathbf{V_C}$ from each cell $C$ of the algebraic set $W$ over the family of polynomials $\calP$.
    \item Run the \pca algorithm for each of the $k$ clusters generated by $\mathbf{V_C}$ to get $k$ affine subspaces.
    \item Return the $k$-affine subspaces that minimize the sum of squares of distances.
\end{itemize}
The degree of $Q$ and polynomials from $\calP$ is at most $4$, $|\calP|=nk^2$, and the real dimension of $W$ is at most $dk(r+1)$ which is the dimension of  $\mathbb{R}^{d \times k (r+1)} \supseteq W$. The algorithm from~\Cref{proposition:alggeom} does at most 
\(
t=(nk^2)^{dk(r+1)}2^{{\mathcal{O}(dk(r+1))}}
\)
operations and produces at most $t$ points. Our algorithm runs $k$ instances of \pca for each such point and hence we have reduced our problem to solving
$k(nk^2)^{dk(r+1)}2^{{\mathcal{O}(dk(r+1))}}=n^{\mathcal{O}(dk(r+1))}$ instances of \pca. Since \pca can be solved in polynomial time the overall running time of our algorithm is in~$n^{\mathcal{O}(dk(r+1))}$.

We can reduce the problem to solving $k(nk^2)^{dk(d-r+1)}2^{{\mathcal{O}(dk(dr+1))}}=n^{\mathcal{O}(dk(d-r+1))}$ instances of \pca using a slightly different characterization of the $k$ $r$-dimensional affine subspaces. 
Note that every $r$-dimensional affine subspace of a $d$-dimensional Euclidean space~$\mathbb{R}^d$  is defined by~$y=\bfp + \mathbf{V}x$, where $\bfp$ is a point in $\mathbb{R}^d$ and $\mathbf{V}\in \mathbb{R}^{d\times r}$ represents the orthonormal basis of an $r$-dimensional linear subspace of $\mathbb{R}^d$. 
 Let $\mathbf{V^C} \in \mathbb{R}^{d\times (d-r)}$ be the orthogonal complement of $\mathbf{V}$.
Thus for every $r$-dimensional affine subspace~$A$ defined by~$y=\bfp + \mathbf{V}x$, we have the corresponding unique affine subspace~$\comp(A)$ defined by~$y=\bfp + \mathbf{V^{C}}x$. 
The points in the algebraic set will now represent $\comp(A)$ instead of the affine subspace $A$.

To this end, we consider the matrix space $\mathbb{R}^{d \times k (d-r+1)}$ and a matrix~$\mathbf{V^{C}}\in \mathbb{R}^{d \times k(d-r+1)}$.
Let~$\mathbf{V^{C}}=\{v_{ij}\}_{i,j}$ and consider the following polynomial conditions:
\begin{align*}
\Tilde{Q}^{O}_{i,j,h}(\mathbf{V^{C}})\coloneqq \sum_{\ell=1}^{d}v_{\ell i}v_{\ell j}=0 \text{ for}&\text{ each } h \in[0,k-1] \text{ and}\\&\text{ each } h(d-r)+1\leq i<j\leq (h+1)(d-r)\\
\Tilde{Q}_i^{N}(\mathbf{V^{C}})\coloneqq \left(\sum_{\ell=1}^{d}v_{\ell i}^2 \right)-1=0 \text{ for}&\text{ each }1 \leq i \leq k(d-r)
\end{align*}
As above, we combine all equations into the following equivalent one:
\[
\Tilde{Q}(\mathbf{V^{C}})= \mathop{\sum}_{\substack{ h(d-r)+1\leq i<j\leq (h+1)(d-r) \\  h \in[k] }}\left(\Tilde{Q}^{O}_{i,j,h}(\mathbf{V^{C}})\right)^2+ \sum_{i=1}^{k(d-r)}\left(\Tilde{Q}_i^{N}(\mathbf{V^{C}})\right)^2.
\]
Again, let~$\Tilde{W}\subseteq \mathbb{R}^{d \times k(d-r+1)}$ be the algebraic set of all matrices satisfying the above equation.
We view each element of $\Tilde{W}$ as a representation of $k$ affine subspaces of dimension~${d-r}$. For every element $\mathbf{V^{C}} \in \Tilde{W}$ and for each~$h \in [0,k-1]$, the columns $h(d-r)+1$ to~$(h+1)(d-r)$ represents the basis of the linear subspace corresponding to the $h$\textsuperscript{th} affine subspace of rank~$d-r$ and column~$k(d-r)+h+1$ represents the corresponding offset point. Let~$col_1,col_2,\ldots col_{d\times k(d-r+1)}$ be the columns of an element~$\mathbf{V^{C}} \in \Tilde{W}$. The $h$\textsuperscript{th} affine subspace is represented by the linear subspace spanned by the columns $col_{h(d-r)+1}$ to~$col_{(h+1)(d-r)+1}$ and the point corresponding to column~$col_{k(d-r)+h+1}$. We again denote the corresponding submatrix by~$\mathbf{B_h}$.
The square of the distance from a point $\bfy \in \mathbb{R}^d$ to the orthogonal complement of an $(d-r)$-dimensional affine subspace defined by~$\bfp + \mathbf{V^{C}}x$ is given by~$\sum_{i=1}^{d-r}(v_i(\bfy-\bfp))^2=\norm{\mathbf{V^{C}}(\bfy-\bfp)}^2$. 
Now we define the family~$\Tilde{\calP}=\{\Tilde{P}_{i,j,f}\}_{1\leq i<j\leq k, 1\leq f\leq n}$ of polynomials over $W$ as 
\[
\Tilde{P}_{i,j,f}(\mathbf{V^{C}})=\norm{\mathbf{B_i}(\bfx_f-col_{k(d-r)+i+1})}^2_F-\norm{\mathbf{B_i}(\bfx_f-col_{k(d-r)+j+1})}^2_F.
\]
Note that~$\Tilde{P}_{i,j,f}(\mathbf{V^{C}})>0$ corresponds to the case that the $j$\textsuperscript{th} $r$-dimensional affine subspace represented by $(\mathbf{V^{C})^C} = \mathbf{V}$ is closer to $\bfx_f$ than the $i$\textsuperscript{th} $r$-dimensional affine subspace represented by $\mathbf{V}$.
By the same argument as above, we can iterate over all non-empty cells of sign conditions and for each cell, we compute a solution by solving~$k$ instances of \pca. 
The degree of $\Tilde{Q}$ and polynomials from $\Tilde{\calP}$ is at most $4$, $|\Tilde{\calP}|=nk^2$ and the real dimension of $W$ is at most~$dk(d-r+1)$ which is the dimension of  $\mathbb{R}^{d \times k (d-r+1)} \supseteq W$. The algorithm from~\Cref{proposition:alggeom} does at most~\(t=(nk^2)^{dk(d-r+1)}2^{{\mathcal{O}(dk(r+1))}}\)~operations and produces at most $t$ points. Our algorithm runs $k$ instances of \pca for each such point and hence we have reduced our problem to~\(k(nk^2)^{dk(d-r+1)}2^{{\Oh(dk(r+1))}}=n^{\Oh(dk(d-r+1))}\)~instances of \pca. The overall running time in this case is $n^{\mathcal{O}(dk(d-r+1))}$. As we can choose the most efficient way to represent the affine subspace the running time of our algorithm is in~$n^{\mathcal{O}(min\{dk(r+1),dk(d-r+1)\})}$. 
\end{proof}
\section{Conclusion}\label{sec:conclusion}
In this work, we showed that \problemLineCl{} is \classW1-hard when parameterized by~$k$ and that \projcl{} is in \classXP when parameterized by both~$k$ and~$d$.
Complementing this result, we also show that the special case \problemHLC{} is \classNP-hard even when~$k=2$ (and it is known to be \classNP-hard for~$d=2$).

We conclude with several open questions. 
Our first question concerns the \classFPT-approxima\-bility of \projcl. To the best of our knowledge, the existence of approximation algorithms with approximation ratio $f(d,k)$ and running time $g(d,k)\cdot n^{\Oh(1)}$ is open for any computable functions~$f$ and~$g$ only depending on~$k$ and~$d$.
The parameterization by~$k$ and~$d$ also remains open.
For the case when $r$ and $k$ are constants, 
Deshpande et al.~\cite{DBLP:journals/toc/DeshpandeRVW06} gave a polynomial-time approximation scheme ({\sf PTAS}) for \projcl.
\iflong
A rich body of literature exists on constructing coresets for specific cases of \projcl{} (\emph{e.g.},~\cite{DBLP:conf/stoc/FeldmanL11,DBLP:conf/fsttcs/VaradarajanX12}). For instance, when the input points have integer coordinates, small coresets can be constructed~\cite{DBLP:conf/fsttcs/VaradarajanX12}, enabling randomized approximation algorithms that run in $f(k)\cdot n^{\mathcal{O}(1)}$ time. 
By combining our \classXP algorithm with such coreset constructions, efficient approximation algorithms can be developed for these cases. 
However, it is known that no \(o(n)\)-sized coreset exists for \projcl{}, even in the special case where \(r = k = 2\) and \({d = 3}\)~\cite{EV05,DBLP:conf/fsttcs/VaradarajanX12}. 
\fi

Another interesting open question concerns the existence of an exact  \classXP algorithm for \projcl{} parameterized by~$d$ and~$k$ for $\ell_p$-norms with~$p \neq 2$. 
Already the base case with a single subspace differs with the case of $p=2$. Clarkson and Woodruff~\cite{ClarksonW15} show that for every $p \in [1, 2)$, the problem of finding an $r$-dimensional subspace $F$ that minimizes (or even $(1 + 1/d^{\Oh(1)})$-approximates) $\sum_{i=1}^n \text{dist}(\bfx_i, F)^p$ is \classNP-hard.
For~$p > 2$, Deshpande et al.~\cite{DBLP:conf/soda/DeshpandeTV11} show \classNP-hardness and UGC-hardness.

 
%



	\bibliographystyle{plainnat}

	\bibliography{book_pc}
\end{document}